\documentclass[12pt]{article}
\usepackage{a4wide}
\usepackage{amssymb}
\usepackage{mathrsfs}
\begin{document}
{\renewcommand{\thefootnote}{\fnsymbol{footnote}}
\hfill  IGC--08/6--7\\
\medskip
\begin{center}
{\LARGE Anomaly freedom in perturbative loop quantum gravity}\\
\vspace{1.5em}

Martin Bojowald, Golam Mortuza Hossain, Mikhail Kagan

\vspace{0.5em}

Institute for Gravitation and the Cosmos,
The Pennsylvania State University,\\
104 Davey Lab, University Park, PA 16802, USA\\

\vspace{1em}

S.~Shankaranarayanan

\vspace{0.5em}

Institute of Cosmology and Gravitation, University
of Portsmouth, Mercantile House, Portsmouth~P01 2EG, U.K.

\vspace{1.5em}
\end{center}
}

\setcounter{footnote}{0}

\newcommand{\proofend}{\raisebox{1.3mm}{\fbox{\begin{minipage}[b][0cm][b]{0cm}
\end{minipage}}}}
\newenvironment{proof}{\noindent{\it Proof:} }{\mbox{}\hfill \proofend\\\mbox{}}
\newenvironment{ex}{\noindent{\it Example:} }{\medskip}
\newenvironment{rem}{\noindent{\it Remark:} }{\medskip}

\newtheorem{theo}{Theorem}
\newtheorem{lemma}{Lemma}
\newtheorem{defi}{Definition}

\newcommand{\case}[2]{{\textstyle \frac{#1}{#2}}}
\newcommand{\lP}{l_{\mathrm P}}
\newcommand{\be}{\begin{equation}}
\newcommand{\ee}{\end{equation}}
\newcommand{\bq}{\begin{eqnarray}}
\newcommand{\eq}{\end{eqnarray}}

\newcommand*{\R}{{\mathbb R}}
\newcommand*{\N}{{\mathbb N}}
\newcommand*{\Z}{{\mathbb Z}}
\newcommand*{\Q}{{\mathbb Q}}
\newcommand*{\C}{{\mathbb C}}

\newcommand{\md}{{\mathrm{d}}}
\newcommand{\tr}{\mathop{\mathrm{tr}}}
\newcommand{\sgn}{\mathop{\mathrm{sgn}}}
\newcommand{\pd}[2][]{\frac{\partial #1}{\partial #2}}
\def\f{\frac}
\def\t{\tilde}
\def\H{{\mathscr H}}
\def\h{{\cal H}}
\def\l{{\cal L}}
\newcommand{\D}{{\partial}}

\begin{abstract}
 A fully consistent linear perturbation theory for cosmology is
 derived in the presence of quantum corrections as they are suggested
 by properties of inverse volume operators in loop quantum
 gravity. The underlying constraints present a consistent deformation
 of the classical system, which shows that the discreteness in loop
 quantum gravity can be implemented in effective equations without
 spoiling space-time covariance. Nevertheless, non-trivial quantum
 corrections do arise in the constraint algebra. Since correction
 terms must appear in tightly controlled forms to avoid anomalies,
 detailed insights for the correct implementation of constraint
 operators can be gained. The procedures of this article thus provide
 a clear link between fundamental quantum gravity and phenomenology.
\end{abstract}

\section{Introduction}

Quantum gravity is expected to play a role in the early universe, in
such a way that it may become subject to observational tests in
cosmology. As is well known, a complete theory of quantum gravity is
difficult to construct, and even if one would have a fully convincing
candidate it would remain difficult to link such a fundamental
formulation to clear-cut observational consequences. Daunting as this
may seem, such a problem is not specific to quantum gravity and has
been circumvented highly successfully in other areas. For instance, to
date there is no complete and fully rigorous construction of
interacting quantum field theories on flat space-time, and yet clear
and experimentally well-tested procedures to extract predictions have
been used for decades.

Quantum gravity certainly does have additional problems which do not
arise in quantum field theories on flat space-time. Paramount among
these issues, related to the general covariance of the theory and thus
to consequences of the fully constrained nature, are the problem of
time, the self-interacting nature of gravity and the
notion of a physical Hilbert space. Yet, quantum field theory is a
lesson that the lack of a completely formulated underlying theory
should not prevent one from making trustworthy statements valid at,
e.g., low energies. Not all the mathematical constructions, whose
well-defined existence one may wish to prove rigorously, are required
for such purposes. They are surely necessary at a fundamental level,
and they have often stimulated much further research. But they do not
directly relate to observables, and thus can, for some purposes, be
ignored.

The key tool for extracting potentially observable consequences
without being paralyzed by open issues in a fundamental framework are
effective formulations. They capture quantum effects by describing
relevant aspects of an evolving wave function. They allow one to focus
on the relevant degrees of freedom, such as expectation values or
fluctuations, rather than whole wave functions and technical issues of
how they may be represented. And if applied carefully enough, they not
only provide reliable self-consistent predictions but also link back
to the full theory where they originate and thus provide fundamental
insights.

For many purposes, almost all the information contained in a wave
function is irrelevant, and a few state parameters of finite number
suffice for all potentially observable consequences that can be
imagined. This is what provides a much more economical derivation of
physical results.  One should note that effective equations are not
merely an amendment of classical equations by quantum correction
terms, although one can always obtain such equations in semi-classical
regimes because the classical limit must be respected. However,
effective equations apply more generally and constitute a systematic
approximation scheme to analyze full quantum properties such as
dynamical expectation values of a state.

Best known among effective formulations are probably low-energy
effective actions in particle physics
\cite{PositronEffAc,VacPolEffAc}. But they are also available in
canonical formulations \cite{EffAc,EffectiveEOM,Karpacz}, where they
have proven fruitful for quantum cosmology
\cite{BouncePert,BounceSqueezed}. They can be extended to constrained
systems, where they provide effective constraints for the state
parameters \cite{EffCons}. Thus, all ingredients are given which are
necessary for an application to loop quantum gravity
\cite{Rov,ALRev,ThomasRev} and a derivation of its effects.  Due to
the complexity of the problem, there is no complete derivation of
effective equations for loop quantum gravity, but several
characteristic quantum effects are known and can be analyzed.  Taken
together, all quantum corrections provide a complicated substitute for
the classical equations, but they can be separated and studied
individually. As we will see in this article, this provides crucial
insights into what should happen in a consistent full theory of
quantum gravity. We will explicitly construct anomaly-free constraints
which incorporate quantum corrections of inverse metric components as
they occur due to the discreteness of loop quantum gravity. A
companion paper \cite{ScalarGaugeInv} will use these considerations of
effective constraint algebras to provide quantum corrected
cosmological perturbation equations in terms of gauge invariant
variables.

\section{Effective equations and effective constraints}

We start by reviewing the scheme of effective equations for canonical
formulations. For an unconstrained system, the quantum theory is given
by a Hamiltonian operator $\hat{H}$ which is self-adjoint on a given
Hilbert space. It determines evolution of states $\psi$ by the
Schr\"odinger equation
\begin{equation}
 i\hbar \dot{\psi}= \hat{H}\psi
\end{equation}
and allows one to solve, e.g., for scattering amplitudes by evolving
given initial states into possible final states.

\subsection{Effective equations of motion}

Alternatively, one can view the expectation value of the Hamiltonian,
$H^Q:=\langle\hat{H}\rangle$, as a functional on the infinite
dimensional space of states. It generates the same Schr\"odinger
evolution by Hamiltonian equations of motion
\begin{equation}
 \frac{\md}{\md t}\langle\hat{O}\rangle =
 \frac{\langle[\hat{O},\hat{H}]\rangle}{i\hbar}=:
 \{\langle\hat{O}\rangle,H^Q\}\,.
\end{equation}
General expressions for the relevant Poisson brackets on the right
hand side can be computed from commutators of basic operators used in
the given quantum theory.

A solution for $\langle\hat{O}\rangle(t)$ to the Hamiltonian equation
of motion has the same information as the expectation value
$\langle\psi|\hat{O}|\psi\rangle$ computed in a state $|\psi\rangle$
satisfying the Schr\"odinger equation (or as the expectation value of
a Heisenberg operator). However, in general $H^Q$, evaluated in a
given state, depends not only on expectation values of basic operators
but also on fluctuations and all other moments of the state. There is
thus a complicated, infinite dimensional coupled system of equations
involving not only the time dependence of $\langle\hat{O}\rangle$ but
also independent quantum variables such as
$\langle\hat{O}^n\rangle\not=\langle\hat{O}\rangle^n$ (or expectation
values of other operators) which appear in quantum theory.

Effective equations are obtained when one can self-consistently
determine regimes where these infinitely many equations can be
decoupled to a finite system. This usually happens in semi-classical
regimes where higher moments of a state are sub-dominant compared to
low ones, but effective equations can be applied more generally. The
1-particle irreducible low-energy effective action, for instance, can
be derived in an approximation consisting of a combination of an
adiabatic expansion with one in $\hbar$ \cite{EffAc,Karpacz}. Some
rare solvable systems can be studied by exact effective equations
without truncation, which in cosmology is realized for a flat
isotropic model sourced by a free, massless scalar \cite{BouncePert}.

\subsection{Effective constraints}

Gravity is governed by constraints rather than a true Hamiltonian.
Just like the Hamiltonian before, constraint operators $\hat{C}_I$
give rise to principal effective constraints
$\langle\hat{C}_I\rangle$, but with those an infinite tower of other
constraints for quantum variables is generated \cite{EffCons}. We are
thus dealing with a system of infinitely many constraints on an
infinite dimensional phase space even for a single classical canonical
pair. The higher constraints can be ignored in our treatment of
characteristic loop quantum gravity effects appearing in the principal
constraints. As we will see, this is sufficient to arrive at a
consistent constraint algebra together with the cosmological
perturbation equations it implies. Higher constraints would be
required if we were interested in the constrained evolution of higher
moments.

The principal constraints are obtained from expectation values of
quantum operators and thus contain several quantum effects. In
general, they depend on quantum variables and include the quantum
back-reaction of, e.g., fluctuations and correlations on expectation
values. But especially in loop quantum gravity they also contain
characteristic effects which are a consequence of the fundamental
quantum representation used. In loop quantum gravity, these translate
consequences of the kinematical discreteness of the loop
representation into effective equations and thus show implications for
dynamical states. The basic reason for properties of the loop
representation is the use of SU(2)-valued holonomies of the
Ashtekar--Barbero connection instead of linear functions of the
connection \cite{LoopRep}. In particular, the basic variables become
complex but still have to obey certain reality conditions to ensure
that the correct number of classical degrees of freedom is quantized.
This is usually implemented by requiring self-adjointness or unitarity
of basic operators in the quantum representation, but one advantage of
effective equations is that reality conditions can be represented
representation independently at the level of expectation values and
quantum variables \cite{BouncePert,EffCons}. This remains true after
solving the effective constraints, thus showing crucial properties of
the physical Hilbert space even in cases where the physical inner
product may be difficult to construct.

A further advantage of effective equations, especially for the purpose
of cosmological perturbation theory, is the issue of introducing a
background geometry to define perturbative expansions in gravitational
variables such as inhomogeneities. While the underlying theory of loop
quantum gravity is background independent in a form which does not
make it straightforward to introduce a perturbative background via
states, effective constraints can easily be expanded by perturbing
expectation values around a background of the desired classical form.
The background then enters via a selection of a class of states to
compute effective constraints \cite{InhomLattice}. (See also
\cite{BoundaryGraviton} for a conceptually similar proposal based on
boundary states.) Using a Friedmann--Robertson--Walker background, for
instance, allows one to derive cosmological perturbation equations
directly --- and effectively --- from a background independent quantum
theory of gravity.

For the resulting set of equations to be meaningful, they must be
consistent in that they derive from constraints which are anomaly-free
also in the presence of quantum corrections. If this fails, it will be
impossible to express the quantum corrected perturbation equations
solely in terms of gauge-invariant variables as they are determined by
the gauge flow of the corrected constraints. In particular, as we will
demonstrate in this article, {\em off-shell anomaly freedom} is
required. (The importance of the off-shell anomaly problem was also
emphasized in \cite{NPZRev} based on alternative fundamental
considerations.) If one has an anomaly-free quantization with
constraint operators such that the constraint algebra
$[\hat{C}_I,\hat{C}_J]$ closes to a first class system, then also the
algebra of principal quantum constraints
$\{\langle\hat{C}_I\rangle,\langle\hat{C}_J\rangle\}$ will close
because it derives from the commutator algebra. (If there are
structure functions, higher constraints as mentioned in the beginning
of this subsection will be involved.) Approximations to effective
constraints then have to be done self-consistently in such a way that
violations of closure of the algebra do not happen up to the order
considered in an expansion.

In loop quantum gravity, however, no satisfactory form of all
constraint operators is known which would satisfy the requirement of
off-shell closure. (The arguments in \cite{AnoFree,QSDI}, for
instance, specifically refer to partially on-shell statements, and
also the reformulation of anomaly-freedom as a condition for the
existence of observables in \cite{Master,AQGI} is on-shell.) Without
off-shell closure, on the other hand, physical applications based on
the usual form of cosmological perturbation equations are impossible.
(While applications may be possible based on a complete set of quantum
observables in a form of reduced quantization, this route does not
seem manageable.)  The final advantage of effective formulations
exploited in this article is then that one can ensure off-shell
closure of the {\em effective constrained system}. Thus, one can
include known quantum effects as they occur in a quantization where
one has not yet taken care of anomaly freedom, obtain candidates for
effective constraints with those corrections in a suitable
parameterized form (reflecting either quantization ambiguities or
incomplete knowledge of properties of a quantum operator), and compute
their Poisson relations. In general, this algebra will exhibit
anomalies, but in some cases one can adapt the correction functions
used in the parametrization of quantum effects such that anomalies
vanish.

If there is no such adaptation, this specific quantum correction would
be ruled out. But if one can successfully remove anomalies while
keeping non-trivial quantum corrections, one will learn how
specifically the quantum effect has to arise in quantum operators, and
how completely quantization ambiguities can be fixed by the
requirement of anomaly-freedom. The advantage of effective equations
is then that one can do such an analysis order by order in various
expansions, instead of having to face a complicated operator algebra
in which all possible quantum effects are included at once. The result
of completing such a program will not only be consistent sets of
equations of motion which can be used for applications, but also
valuable feedback on the underlying fundamental theory which in our
case will be loop quantum gravity. Thus, we are providing a clear link
between fundamental properties of quantum gravity and its
phenomenology.

\subsection{Cosmological perturbation equations}

Linearization of Einstein's equations around
Friedmann--Robertson--Walker (FRW) space-times provides cosmological
perturbation equations for ten metric components. These metric
perturbations are subject to coordinate (gauge) transformations
parametrized by an infinitesimal 4-vector field $\xi^\mu$
($\mu=0,\ldots,3$), which, in presence of matter, generically give
rise to six {\em gauge invariant} perturbations, i.e.\ combinations of
metric and matter perturbations which remain unchanged under linear
changes of coordinates. In the linear regime, the former
decouple into three independent modes: scalar,
vector, and tensor, carrying two degrees of freedom each. The
evolution of vector and tensor modes taking into account corrections
expected from loop quantum gravity was investigated in \cite{vector}
and \cite{tensor} respectively. In this paper, we focus on the scalar
perturbations which, along with the background FRW-metric, take the
form
\be\label{MetricPert} \md s^2 =
a^2(\eta)\left(-(1+2\phi)\md \eta^2 +2\D_a B\md \eta \md x^a -
((1-2\psi)\delta_{ab}+2\D_a\D_b E)\md x^a\md x^b\right)\,,
\ee
where the scale factor $a$ is a function of the conformal time $\eta$
and the spatial indices run from 1 to 3. The perturbations $\phi$,
$\psi$, $B$ and $E$ are then combined into the two gauge invariant
Bardeen potentials \cite{Bardeen}
\begin{eqnarray}\label{BardeenVars}
\Psi&=&\psi-\H\left(B-\dot E\right)\\
\Phi&=&\phi+\left(B-\dot E\right)^{\dot{}}+\H\left(B-\dot
E\right).\nonumber
\end{eqnarray}
whose evolution is governed by the linearized Einstein 
equations
\cite{CosmoPert}
\begin{eqnarray}
&&\nabla^2\Psi-3\H (\H \Phi + \dot\Psi)=-4\pi G a^2
\delta T_0^{0({\rm GI})} \label{PertI}\\
&&\D_a\left(\H \Phi + \dot \Psi \right)=-4\pi G a^2
\delta T_a^{0({\rm GI})} \label{PertII}\\
&&\left(\ddot \Psi + \H (2\dot\Psi+\dot\Phi)+(2\dot\H+\H^2)\Phi+
\f{1}{2}\nabla^2(\Phi-\Psi)\right)\delta^b_a\nonumber\\
 &&\qquad-\f{1}{2}\D^b\D_a(\Phi-
\Psi)=4\pi G a^2 \delta T_a^{b({\rm GI})}\,. \label{PertIII}
\end{eqnarray}
Here a dot denotes derivative w.r.t conformal time, $\H\equiv \f{\dot
a}{a}$
is the conformal Hubble parameter and
$\delta T^{({\rm GI})}$ are gauge invariant perturbations of the
matter stress-energy tensor. These equations are commonly derived from
the covariant field equations or by varying an action expanded to
second order in the linearized fields. But Hamiltonian formulations
exist for the same procedure, which is more suitable for a comparison
with canonical quantum gravity (in particular in Ashtekar variables as
used in \cite{HamPerturb}).

To formulate the Hamiltonian setting, the action is used to determine
Poisson brackets, and thus a decomposition into configuration fields
and their momenta, as well as constraint functions. The constraints
serve several purposes: (i) they restrict initial values of the fields
to those allowed values which make the constraints vanish, (ii) they
generate gauge transformations which in the case of general relativity
agree with coordinate transformations, and (iii) they provide
equations of motion for the fields in any coordinate time
parameter. (The latter is itself subject to the coordinate changes by
transformations generated by the constraints.) All this is necessary
to reproduce a covariant system even though distinguishing momenta,
which are related only to time but not space derivatives of fields,
invariably removes manifest covariance from the Hamiltonian formalism.

For this to be consistent, it is crucial that the constraints are
preserved under the time evolution they generate. This is
automatically guaranteed if they form a first class set, i.e.\ a set
of functionals whose mutual Poisson brackets vanish when evaluated in
fields satisfying the constraints. In other words, the gauge
transformations and evolution generated by the constraints then define
vector fields on field space which are tangent to the sub-manifold
defined by the vanishing of constraints. Starting on the constraint
surface, either changing the gauge or following evolution will then
keep us on the constraint surface. This is certainly realized
classically, as a reflection of the general covariance of the
underlying theory.

However, if quantum aspects are implemented, one must ensure that this
consistency requirement remains maintained: the quantization must be
anomaly-free. Otherwise the equations may show the wrong type and
number of degrees of freedom if formerly gauge quantities acquire
gauge invariant meaning. Or, worse, anomalies may make the equations
inconsistent to the degree that no non-trivial solution exists at
all. Anomaly-freedom is thus a key requirement not only for the
consistency of an underlying fundamental theory but also for the
possibility of applications. Quantum corrections cannot appear in
arbitrary forms, but only in restricted ways such that the constraints
form a closed algebra under Poisson brackets. In particular,
anomaly freedom will reduce some of the
arbitrariness of the form of loop quantum gravity corrections.

Moreover, as we will see explicitly, to provide quantum corrections to
Eqs.~(\ref{PertI}), (\ref{PertII}) and (\ref{PertIII}) the algebra
must close off-shell, i.e.\ it is not enough that the Poisson brackets
of constraints vanish when the constraints are satisfied but even on
parts of the phase space where constraints $C_I$ do not vanish we must
produce a closed algebra of a form $\{C_I,C_J\}= f^{K}_{IJ}(A,E)
C_K$. (Here, $A$ and $E$ denote the canonical fields which may appear
in the coefficients of the algebra; this means that in general we have
structure functions rather than structure constants.) The effective
algebra may differ from the classical one, and thus be quantum
corrected as well as the constraints; but it must still close. The
reason is that the whole set of coupled equations must be consistent,
which presents a mixture of constraint equations (\ref{PertI}) and
(\ref{PertII}) and evolution equations given by (\ref{PertIII})
together with the continuity or Klein--Gordon equation. To ensure that
these equations are consistent, we must consider the constraints
before they are solved. Consistency then requires an off-shell closure
of the constraint algebra. Practically, the consequence is that only
in this case we can express all the equations solely in terms of
gauge-invariant variables as they are determined by the quantum
corrected constraints. Once this is achieved, the equations are
consistent and can be solved and analyzed. In the absence of off-shell
closure, on the other hand, there would be left-over terms in the
equations of motion which contain gauge-dependent quantities making
such an evolution unphysical.

In this context, it is important to realize that there is no shortcut
to implementing the quantum corrections of fully perturbed field
equations consistently. (Notwithstanding the fact that this has been
attempted on numerous occasions such as
\cite{PowerLoop,PowerPert,InhomEvolve} in the context of loop quantum
cosmology, including by some of the present authors.\footnote{In
  particular, in Ref.~\cite{InhomEvolve} a gauge fixing choice was
  made which, as follows from the present paper, resulted in
  anomalies.  Also the recent proposal \cite{LQCStepping} of building
  inhomogeneous configurations from small numbers of patches
  subdividing space suffers from anomalies if more than two patches
  are considered.  In this work, the authors also use a
  Born-Oppenheimer approximation --- by assuming that the
  wave-function of the background metric and of the perturbations can
  be separated. This approximation is valid only if the wavelength of
  the perturbations $\lambda$ are much larger than the Planck length,
  i.e.\ $\lambda \gg l_P$. In the present work, such an approximation
  is not needed to derive the perturbation equations.})  Consistency
even for the purposes of phenomenological applications is intimately
linked to the fundamental problem of anomaly-freedom once
inhomogeneities enter the game. In homogeneous models of gravity there
is just one constraint, which clearly has a vanishing Poisson bracket
with itself and thus forms an off-shell closed algebra. Thus, in
homogeneous quantum cosmology there is no anomaly problem whatsoever.
Here, quantum corrections can be implemented at will, only restricted
by possible self-imposed conditions such as the desire to be as close
to a candidate for a ``full,'' non-symmetric theory as possible as it
is expressed in loop quantum cosmology \cite{LivRev}. (Some of the
structures, chiefly the kinematical quantum representation, of loop
quantum cosmology can be linked to loop quantum gravity and are thus
more restricted
\cite{InhomLattice,SymmRed,SymmStatesInt,Reduction,Rieffel}. But no
such derivation exists yet for the constraints which are most
important to see the precise role of quantum corrections on the
dynamics.)

It is then sometimes proposed to implement quantum corrections only in
the background evolution, for instance by effects motivated from loop
quantum gravity, and then use some inhomogeneous degree of freedom
such as a matter field as a measure of perturbations around the
background. If just the background is quantized, one knows corrections
only in its evolution but would have to keep the structure of
classical perturbation equations otherwise unchanged. This is rarely
consistent, and the treatment is not gauge-invariant. Gauge invariant
quantities in general relativity such as (\ref{BardeenVars}) combine
several metric perturbations and possibly matter fields. Taking only a
matter field, say, as the measure of perturbations means that one is
fixing the gauge (by implicitly assuming non-gauge invariant metric
perturbations to vanish) without even knowing what the gauge
transformations are.  The classical case of linear perturbations
around Friedmann--Robertson--Walker spacetimes allows gauges where
only the matter fields are inhomogeneous but not the metric
like,
for instance, the uniform density
gauge. However, quantum corrections change the constraints and thus
the gauge transformations they generate. The form and availability of
certain gauges changes, and it is no longer possible to re-express the
gauge-fixed results in terms of the gauge invariant quantities unless
one considers the full gauge problem. This can only be done when
initially all perturbations are allowed and the anomaly-issue is faced
head-on.

In the context of classical cosmological perturbations, the above
arguments can be rephrased in the following manner. The effective
corrections arising in loop quantum cosmology can formally be written
as:
\begin{equation}
G_{\mu\nu} = 8 \pi G \left(T_{\mu\nu} + \tau_{\mu\nu}\right)
\end{equation}
where $\tau_{\mu\nu}$ contains all the corrections from loop
quantum 
cosmology and $T_{\mu\nu}$ corresponds
to the stress-tensor of the classical matter field. Although the
matter field might be an ideal fluid, the stress-tensor
$\tau_{\mu\nu}$ arising due the new physics cannot necessarily be
treated as a perfect fluid. More importantly, the perturbation of
the stress-tensor $\delta \tau_{ab}$ will in general have some
anisotropic stress and the velocity perturbation $\delta v_{a}$
will not vanish for a standard gauge choice. Hence, it is
important to study the perturbations in a gauge invariant manner.

If there is an anomaly-free version of quantum corrected constraints
and the corresponding form of covariance, one could compute complete
gauge-invariant quantities to arbitrary orders in an expansion by
inhomogeneities; see e.g.  \cite{NonLinPert,PertObsI,PertObsII} or, in
deparameterized form after introducing dust as a clock matter system,
\cite{BKdustI,BKdustII}. As discussed, this requires off-shell anomaly
freedom of the constraints which is not easy to realize in closed
form. The treatment by effective constraints then provides a key
advantage: one can verify anomaly-freedom order by order in the
expansion by inhomogeneities (which may be combined with a
semiclassicality expansion in $\hbar$). This can be done with much
more ease than a full anomaly analysis but still, as we will see
explicitly, provides crucial feedback for the full theory.

\subsection{Correction functions}

Any quantization, such as loop quantum gravity, implies characteristic
effects which change the classical behavior. Almost always, there are
quantum back-reaction effects by state parameters such as fluctuations
and correlations on expectation values. (The only exceptions are free
or solvable models such as the harmonic oscillator where moments of a
state evolve independently of expectation values.) In addition, the
specific quantum representation may imply further characteristic
effects, which in the case of loop quantum gravity are all related to
the spatial discreteness of its kinematical representation. The
classical set-up makes use of basic variables given by a densitized
triad $E^a_i$, which provides the spatial metric via $E^a_iE^b_i=\det
(q_{cd}) q^{ab}$, and the Ashtekar connection $A_a^i=\Gamma_a^i+\gamma
K_a^i$ with the spin connection $\Gamma_a^i$ and extrinsic curvature
$K_a^i=E^b_iK_{ab}/\sqrt{|\det (E^c_j)|}$. This canonical pair of
fields is then quantized in the form of fluxes, i.e.\ integrations of
the triad over surfaces, and holonomies or parallel transports of the
connection. The resulting background independent representation has
characteristic properties of spatial discreteness such as a discrete
spectrum of flux operators (which contains zero).  Such properties
imply associated quantum corrections which appear whenever there are
inverse powers of densitized triad components (some of which would
classically diverge at singularities) or holonomies of a connection
rather than just connection or curvature components.

All these corrections typically occur at the same time and must be
combined in a complete treatment. While one type of correction might
be dominant in certain regimes, this would not be known a priori but
had to be shown by a dedicated analysis. Nevertheless, due to the
complexity of general quantum corrections, it is legitimate to
separate the different corrections at first, analyze individual
effects and then combine results. In spirit, this is similar to the
calculation of corrections to an atomic spectrum, which can be done
individually for relativistic corrections, spin-orbit interaction
etc.\ and eventually combined in a complete spectrum. In this paper,
we focus on loop quantum gravity corrections as they arise from inverse
components of the triad.

These corrections are already relevant for cosmology, where they have
been analyzed in preliminary forms in homogeneous and inhomogeneous
contexts
\cite{InflationWMAP,Robust,PowerLoop,ScaleInvLQC,CCPowerLoop,SuperInflLQC,RefinementMatter,RefinementInflation,TensorBackground,TensorRelic}. (Note
that a sub-dominance of these corrections compared to those due to
holonomies has been claimed based on an analysis of isotropic models
\cite{APS}. However, this is based largely on an inadvertent and
artificial suppression in the models used \cite{SchwarzN}; see also
App.~\ref{a:iso}. In any case, the arguments put forward in the
context of \cite{APS} do not apply to inhomogeneous situations.) The
precise form of such corrections as they would result in a principal
constraint from the expectation value of a constraint operator cannot,
at present, be computed due to the complicated form of the volume
spectrum which would be required (see e.g.\ \cite{BoundFull}). But
partially the behavior is known as it follows for instance for
diagonal triads \cite{QuantCorrPert}. The typical behavior is that the
classical function of triad components is multiplied by a correction
function $\alpha(E^a_i)$ which approaches the classical expectation
one for large values of triad components.  (More precisely, the
function depends on fluxes, i.e.\ triad components integrated over
elementary plaquettes of a discrete quantum state. This makes the
functional behavior independent of the choice of coordinates.)

At smaller scales, however, the function starts to deviate from one
and implies quantum corrections. If the correction function is
evaluated on an isotropic background \cite{InvScale}, it has a peak at
a certain characteristic scale $a_*$ of height larger than one, and
then drops off at smaller scales to reach zero for vanishing
triads. Notice that inhomogeneous contexts and states make it
meaningful to speak about this behavior in terms of the scale factor
$a$. In exactly isotropic models which are spatially flat, the
absolute size of the scale factor has no meaning. However, the
argument of correction functions is determined by a dimensionless
ratio given by $q:=\ell_0^2a^2/\ell_{\rm P}^2=: a^2/a_*^2$ where
$\ell_0$ is the size in coordinates of an elementary plaquette whose
flux appears as an argument. The product $\ell_0a$ has unambiguous
meaning because it does not change under rescaling coordinates (which
would change both $a$ and $\ell_0$ individually). The peak of the
correction functions occurs near $q\sim 1$, i.e.\ $a\sim a_*$. The
characteristic scale $a_*$ can be written as $a_*=\ell_{\rm P}/\ell_0=
({\cal N}/V_0)^{1/3}\ell_{\rm P}$ where ${\cal N}$ is the number of
vertices of an underlying state contained in a region of coordinate
volume $V_0$. The ratio ${\cal N}/V_0$ appearing in $a_*$ is thus the
patch density of an underlying discrete state measured in a given
coordinate system. For nearly homogeneous configurations, it does not
depend on the region or on $V_0$, but on coordinates.  (The physical
vertex density which would be independent of coordinates is ${\cal
N}/(a^3V_0)$, but it would not be appropriate to determine a
characteristic scale for $a$. Note that near $a\sim a_*$ there is one
patch per Planck cube; upper bounds for the patch density can be
obtained from phenomenological considerations, such as from big bang
nucleosynthesis \cite{FermionBBN}.)  The value of $a_*$ depends on the
normalization of the scale factor. But it also depends on the vertex
density which can be large. Thus, the peak of correction functions for
a denser state is realized on larger scales, which increases the
corresponding quantum corrections.

An additional implication of the appearance of the vertex density is
that ${\cal N}$ is typically history dependent
\cite{InhomLattice,CosConst} if the dynamical quantum evolution
refines the state as the universe grows (rather than just blowing up a
fixed lattice). Thus, also the scale $a_*$ is history dependent which
contributes to the regime dependence of this type of correction. For a
given background, the history dependence can always be expressed as an
$a$-dependence, which is sometimes seen as problematic because the
scale factor is not coordinate independent. However, given that the
origin of the refinement lies in the inhomogeneous setting, a proper
reduction introduces the correct scaling dependence via additional
parameters depending on the state; see also App.~\ref{a:iso}.

For an implementation of perturbative inhomogeneities, regimes where
relevant scales fall below $a_*$ pose difficulties because the scale
of inhomogeneity would be close to the discreteness scale. In this
paper, we thus assume that scales of the densitized triad are above
the characteristic scale $a_*$, where correction functions deviate
from one by terms perturbative in the Planck length:
\begin{equation}\label{AlphaHomo}
 \alpha(a)= 1+c_{\alpha}\left(\frac{\ell_P^2}{a^2}\right)^{n_{\alpha}}+\cdots
\end{equation}
with positive coefficients $c_{\alpha}$ and $n_{\alpha}$. Both
coefficients can be derived from a specific quantization but are
subject to quantization ambiguities. The coefficient $c_{\alpha}$, in
particular, is then related to $a_*$ (and to $\ell_0$, providing the
correct coordinate dependence in the presence of the scale factor).
Thus, $c_{\alpha}$ may itself depend on $a$ if the vertex number ${\cal N}$ in
a fixed volume, and thus $a_*$, changes with the universe expansion.
We are assuming that the dominant $a$-dependence is via a power-law of
the given form.

Constraints for linearized perturbations will not only require the
dependence of $\alpha(E^a_i)$ on the triad when the latter is
diagonal, but also the dependence on off-diagonal components.
Classically one can always gauge the triad to be diagonal, but gauge
transformations are quantum corrected and a consistency analysis of
the equations must be done before a gauge is fixed. The off-diagonal
dependence of $\alpha$ is not known in explicit form, and it is
difficult to derive because unlike the diagonal case it requires
non-Abelian features of the quantum theory
\cite{BoundFull,DegFull}. As we will see, the consistency analysis of
constraints then relates the off-diagonal dependence to the diagonal
dependence via the condition of anomaly-freedom. Moreover, other terms
in the constraints, including matter Hamiltonians, will also require
characteristic quantum corrections which, in contrast to the primary
correction functions, may not obviously be expected from explicit
quantizations in homogeneous models. Nevertheless, such additional
corrections, called {\em counterterms} in what follows, are required
for anomaly-freedom. In this way, they are fixed in terms of the
primary correction functions depending on diagonal triads. All this not
only provides consistent equations ready to be applied in cosmology,
but also precise feedback on what terms a full anomaly-free quantum
constraint must contain. As technical control on the full setting
increases, these predictions will provide strong consistency checks of
the whole framework.



\section{Canonical perturbation theory and primary correction functions}

In this main part of the paper we develop the ingredients of a
consistent perturbation theory in the presence of quantum corrections
to the classical constraints.

\subsection{Constraints and primary corrections}
\label{PrimCorr}

We first introduce the constraints and primary correction functions
which are expected to arise in the effective constraints. Formally,
the corrections are introduced as multiplicative factors of some terms
in the constraints which depend on the phase space variables and
approach unity in the classical limit. In this paper we restrict
ourselves to correction functions resulting from the quantization of
inverse-triad terms of the constraints. For the primary corrections,
the functions are also assumed to depend only on the triad and to be
local, i.e.\ independent of spatial derivatives of the triad. This
reflects properties of these functions as they have been introduced in
homogeneous models. The input can thus be used to formulate an initial
expectation of the form of such functions. Moreover, in this section,
we assume that the corrections can in principle be obtained from the
full (non-perturbative) theory, and hence should depend only on the
full triad $E^a_i\equiv\bar E^a_i+\delta E^a_i$ rather than on the
background $\bar E^a_i$ and perturbations $\delta E^a_i$ as distinct
arguments. Later on we shall analyze the consistency of such
assumptions. Anomaly-freedom will generate additional counter-terms of
further corrections, which can be re-interpreted as a connection
dependence or non-locality of the primary corrections. Such a
dependence is in any case expected for covariant corrections which
can, e.g., be formulated as functionals of curvature invariants. Of
course, we could put in such a dependence from the outset, but it
would make the calculations much less tractable.

General relativity in Ashtekar variables is subject to the Gauss,
diffeomorphism and Hamiltonian constraints. The Gauss constraint is
identically satisfied for scalar modes and does not need to be
considered here.  In the full quantum theory, the diffeomorphism
constraint does not receive quantum corrections but the Hamiltonian
constraint does \cite{QSDI}. The diffeomorphism constraint is thus
taken as the classical one, $D[N^a]=D_{\rm grav}[N^a]+D_{\rm
matter}[N^a]$ with a gravitational part
\begin{equation} \label{DiffeoConstraint}
D_{\rm grav}[N^a] := \frac{1}{8\pi G\gamma}\int_{\Sigma}\mathrm{d}^3xN^a
\left[(\partial_a A_b^j - \partial_b A_a^j)E^b_j - A_a^j
\partial_b E^b_j \right] ~
\end{equation}
and a matter part
\begin{equation}
D_{\rm matter}[N^a]=\int_{\Sigma}\mathrm{d}^3x N^a \pi \partial_a \varphi
\end{equation}
for a scalar field $\varphi$.

We express the classical gravitational Hamiltonian as
\begin{equation} \label{HamConstClass}
H_{\rm grav}[N] = \frac{1}{16\pi G} \int_{\Sigma}\mathrm{d}^3x N {\mathcal
H}
\end{equation}
in terms of the Hamiltonian density
\begin{equation}\label{Hg_def}
\h=\frac{E_i^a E^b_j}{\sqrt{|\det
E|}}\left(F_{ab}^k{\epsilon^{ij}}_{k}-2(1+\gamma^{-2})
(A-\Gamma)_a^{[i}(A-\Gamma)_b^{j]}\right),
\end{equation}
where the curvature of the Ashtekar connection is given by
\[
F_{ab}^k=2\D_{[a}A_{b]}^k+{\epsilon_{ij}}^{k} A_a^i A_b^j\,,
\]
$\gamma$ is the Barbero--Immirzi parameter and the spin connection
$\Gamma_a^i$ is considered as a functional of the densitized triad
(written explicitly in Eq.~(\ref{SpinConnection})). The presence
of an inverse of the triad determinant, whose quantization does not
have a direct inverse because it has a discrete spectrum containing
zero, suggests the presence of a primary correction function of
inverse-triad type,
\begin{equation} \label{HamConstQuant}
H_{\rm grav}^P[N] = \frac{1}{16\pi G}\int_{\Sigma}\mathrm{d}^3x N
\alpha(E^a_i) {\mathcal H} = H_{\rm grav}[\alpha N] =: H_{\rm
grav}[\tilde{N}] ~.
\end{equation}
For the same reason, primary quantum corrections $\nu(E_i^a)$ and
$\sigma(E_i^a)$ are introduced into the matter part of the Hamiltonian
constraint as
\begin{equation}
H^P_{\rm matter}[N]=\int_{\Sigma}\mathrm{d}^3x N
\left({\nu\h_\pi+\sigma\h_\nabla+\h_\varphi}\right),
\end{equation}
where
\begin{equation}\label{Hm_def}
\h_\pi=\frac{\pi^2}{2 \sqrt{|\det{E}|}},\quad
\h_\nabla=\frac{E^a_i E^b_i
\partial_a \varphi \partial_b
\varphi}{2\sqrt{|\det{E}|}},\quad\h_\varphi=\sqrt{|\det{E}|}V(\varphi)
\end{equation}
are the (classical) Hamiltonian densities corresponding to the
kinetic, gradient and potential terms respectively. Note that the
potential term is not expected to acquire primary quantum corrections
because there is no inverse of the triad in this term. Thus, the
correction does not simply amount to a rescaling of the lapse function
even if $\nu$ and $\sigma$ would equal $\alpha$.

In App.~\ref{a:Unpert} we discuss the Poisson brackets between unperturbed
$H^P[N]\equiv H_{\rm grav}^P[N]+H_{\rm matter}^P[N]$ and $D[N^a]\equiv D_{\rm
  grav}[N^a]+D_{\rm matter}[N^a]$, as well as between $H^P[N_1]$ and
$H^P[N_2]$.

\subsection{Perturbations}

Consider first an action which depends only on one scalar field
$\varphi$.  Generalizations to arbitrary tensor fields will be
considered in the following subsection and in App.~\ref{a:Poisson}.
After the Legendre transform the action takes the form
\be\label{Action}
S[\varphi]=\int{\!\md^4 x \left(\pi
\dot\varphi-\h(\varphi,\pi)\right)},
\ee
where $\h$ is the Hamiltonian density and $\pi$ is the field
momentum. Given a space-time slicing by constant-time surfaces, we
split the fields $\varphi$ and $\pi$ into their homogeneous parts
\be\label{Def_Hom}
\bar\varphi:=\f{1}{V_0}\int{\!\md^3 x\,\varphi}, \quad
\bar\pi:=\f{1}{V_0}\int{\!\md^3 x\,\pi}
\ee
and the inhomogeneous remainder
\be\label{Def_Inhom}
\delta\varphi:=\varphi-\bar\varphi, \quad \delta\pi:=\pi-\bar\pi.
\ee
Here, $V_0$ is the volume of a spatial slice if it is closed, or can
be thought of as a very large (but finite) infrared cutoff volume
otherwise. The coordinate size $V_0$ will only appear in basic
variables and their symplectic structure, but not in final equations
of motion.

We also require the inhomogeneities $\delta \varphi$ and $\delta
\pi$ to be small:
\be\label{Small_Pert}
\left|\f{\delta \varphi}{\bar\varphi}\right| \ll 1, \quad
\left|\f{\delta \pi}{\bar\pi}\right| \ll 1
\ee
for the slicing we use, so that this can be considered as
perturbations around homogeneous solutions.  For a generic
Hamiltonian, these conditions may at some point be violated during the
evolution. In fact, only a narrow class of Hamiltonians admits such a
splitting, for which the inhomogeneities remain smaller than the mean
fields at all times. At the moment, as is common in cosmology of the
early universe, we merely assume that (\ref{Small_Pert}) holds for the
regime under consideration. Hence, from now on we shall refer to
$\delta\varphi$ and $\delta\pi$ as perturbations and will also speak
of the first, second, etc.\ perturbative order, denoted by
superscripts $(1)$, $(2)$, \ldots\ in what follows.  Specifically, we
will be interested in the perturbations up to the second order in the
Hamiltonian, which implies a linear perturbation theory in terms of
equations of motion.

{}From the very definition of $\delta\varphi$ and $\delta\pi$ it
follows that any first-order quantity averages to zero. In
particular,
\be\label{2nd_class_constraints}
\chi_1:=\int{\md^3 x\lambda_1 \delta\varphi}=0, \quad
\chi_2:=\int{\md^3 x\lambda_2 \delta\pi}=0,
\ee
where $\lambda_1$ and $\lambda_2$ are `smearing' constants.%
\footnote{Obviously, the constant factors $\lambda_1$ and $\lambda_2$
  are not necessary in the constraints. They have been introduced to
  make a more explicit connection with the case of tensor fields,
  where such factors would be necessary to contract spatial and
  internal indices} Therefore the first term in the action
(\ref{Action}) splits into two parts:
\be\label{New_Sim_Str}
\int{\md^4 x \pi \dot\varphi} \equiv  \int{\md^4 x
(\bar\pi+\delta\pi)(
\dot{\bar\varphi}+\delta\dot\varphi)}=V_0\int{\md t \bar\pi
\dot{\bar\varphi}}+\int{\md^4 x \delta\pi \delta\dot\varphi},
\ee
yielding the basic Poisson brackets
\be\label{PB}
\left\{\bar\varphi,\bar\pi\right\}=\f{1}{V_0}, \quad
\left\{\delta\varphi(x),\delta\pi(y)\right\}=\delta^3(x-y)
\ee
and, for phase space functions,
$\{,\}:=\{,\}_{\bar\varphi,\bar\pi}+\{,\}_{\delta\varphi,\delta\pi}$,
where
\bq \label{PB_full}
\{F,G\}_{\bar\varphi,\bar\pi}&=&\f{1}{V_0}\left(\f{\D F}{\D
\bar\varphi}\f{\D G}{\D
\bar\pi}-\f{\D F}{\D \bar\pi}\f{\D G}{\D \bar\varphi}\right), \nonumber\\
\{F,G\}_{\delta\varphi,\delta\pi}&=&\int{\md^3 x\left(\f{\delta
F}{\delta(\delta \varphi)}\f{\delta G}{\delta(\delta
\pi)}-\f{\delta F}{\delta(\delta \pi)}\f{\delta G}{\delta(\delta
\varphi)}\right)}\,.
\eq
As discussed in App.~\ref{a:Poisson}, these brackets are not fully
general and in some cases care may be required, but they are
sufficient for calculations done here.

\subsection{Perturbed constraints}
\label{Pert_Constr_Alg}

So far we have used the Ashtekar connection $A_a^i$ as one of the
canonical variables as required for holonomies. From now on we will
explicitly use the spin connection and extrinsic curvature,
\[
A_a^i=\Gamma_a^i+\gamma K_a^i,
\]
where $\gamma$ is the Barbero--Immirzi parameter. Also the canonical pair
$K_a^i=\bar K_a^i + \delta K_a^i,E^a_i=\bar E_i^a + \delta E_i^a$ can
be split into the homogeneous parts
\[
 \bar K_a^i = \bar k \delta_a^i, \quad \bar E_i^a = \bar p
 \delta_i^a,
\]
corresponding to the flat FRW-background, and the inhomogeneous
perturbations which, for the scalar mode, are described by a pair
of scalar functions each:
\begin{equation}\label{ScalarPerturbations}
\delta K_a^i = \delta_a^i \kappa_1 + \partial_a \partial^i
\kappa_2, \quad \delta E_i^a = \delta_i^a \varepsilon_1 +
\partial_i \partial^a \varepsilon_2\,.
\end{equation}
Note that in the perturbed context, the independent phase space
variables are $(\bar k,\bar p)$ and $(\delta K_a^i,\delta E_i^a)$, and 
the non-trivial Poisson brackets between them are given by 
\footnote{See also App.{\ref{a:Poisson}} for details.}
\be\label{BasicPBGrav}
\left\{\bar k,\bar p\right\}=\f{8 \pi G}{3V_0}, \quad
\left\{ \delta K_a^i(x) , \delta E_j^b(y) \right\}= 
8\pi G \delta^i_j \delta_a^b\delta^3(x-y) ~.
\ee
The Hamiltonian density (\ref{Hg_def}), expressed in terms of the
extrinsic curvature, becomes
\begin{equation}\label{Hg_def_1}
\h=\epsilon_{i}^{jk}\frac{E_j^c E_k^d}{\sqrt{|\det E|}}\left[2
\D_c\Gamma_d^i+\epsilon^i_{mn}(\Gamma_c^m\Gamma_d^n-K_c^m
K_d^n)\right]+\h_{\gamma},
\end{equation}
where the last term is proportional to the Barbero--Immirzi parameter,
\[
\h_{\gamma}=\gamma\epsilon_{i}^{jk}\frac{2 E_j^c E_k^d}{\sqrt{|\det E|}}\left(
\D_c K_d^i+\epsilon^i_{mn}\Gamma_c^m K_d^n\right)\equiv
\gamma\epsilon_{i}^{jk}\frac{2 E_j^c E_k^d}{\sqrt{|\det E|}} D_c K_d^i=
2\gamma \frac{E^c_j}{\sqrt{|\det E|}} D_c {\cal G}^j
\]
with the Gauss constraint ${\cal G}^J$, which thus vanishes. Indeed, the
Gauss constraint implies that the extrinsic curvature can be written
as $K_d^i=K_{db}E^{bi}/\sqrt{|\det{E}|}$ where $K_{dc}=K_{cd}$.
Consequently,
\[
\h_{\gamma}\propto\epsilon_{ijk}\frac{E_j^c E_k^d E^b_i}{\det E}
D_c K_{db}=\epsilon^{bcd}D_c K_{db}=0.
\]
Thus the classical theory in $(K_a^i,E^b_j)$ is explicitly insensitive
to the Barbero--Immirzi parameter, as it should. The
$\gamma$-dependence, however, will appear in the correction functions
resulting from a quantization procedure (after which no unitary
transformation exists to change $\gamma$ without leaving a trace on
observable quantities).

The remaining part of the Hamiltonian density can be expanded
straightforwardly in a perturbation series, although the
spin-connection requires some care. Its full expression is
\begin{equation}\label{SpinConnection}
\Gamma_a^i=-\frac{1}{2}\epsilon^{ijk} E_j^b \left(\partial_a
E_b^k-\partial_b E_a^k+E_k^c E_a^l \partial _c E_b^l - E_a^k
\frac{\partial_b (\det E)}{\det E}\right)\,,
\end{equation}
where $E^i_a$ with a lower spatial index designates a co-triad of
density weight minus one, whose perturbed expression reads
\[E_a^l=\frac{1}{\bar p}
\delta_a^l-\frac{1}{\bar p^2}\delta E^c_k \delta_{ca}\delta^{kl}.
\]
The first order part of (\ref{SpinConnection}),
\begin{equation}\label{SpinConnectionLin}
\delta \Gamma_a^i=\frac{1}{2\bar p}\left(
\epsilon_c^{ij}\delta_a^b-\epsilon_c^{ib}\delta_a^j+\epsilon^{ijb}
\delta_{ac}+\epsilon_a^{ib}\delta_c^j\right) \partial_b\delta
E_j^c\,,
\end{equation}
is simplified significantly for a scalar perturbation of the form
(\ref{ScalarPerturbations}).  The diagonal part of $\delta E_j^c$ (in the term
$\D_b\delta E^c_j$, which is $\delta^c_j \partial_b\varepsilon_1$,)
contributes to the linearized spin connection
\[
\delta \Gamma_a^{i(\rm diag)}= \frac{1}{2\bar p}\left(0-\epsilon_a^{ib}\D_b \varepsilon_1-
\epsilon_a^{ib}\D_b\varepsilon_1+3\epsilon_a^{ib} \D_b \varepsilon_1\right)\\
=\frac{1}{2\bar p}\epsilon_a^{ib}\D_b \varepsilon_1\,.
\]
For the off-diagonal perturbation, on the other hand, the
expression $\D_b \delta E^c_j\equiv \D_b\D^c\D_j \varepsilon_2$ is
symmetric in the indices $b,c,j$, implying that only one term in
(\ref{SpinConnectionLin}) remains:
\[\delta \Gamma_a^{i(\rm off-diag)}=\frac{1}{2\bar p}\left[
0-0+0+\epsilon_a^{ib} \D_b \Delta \varepsilon_2 \right]\\
=\frac{1}{2\bar p}\epsilon_a^{ib}\D_b \Delta \varepsilon_2.
\]
Combining the last two expressions, we obtain the
linearized spin connection
\begin{equation}
\delta \Gamma_a^i =
\delta \Gamma_a^{i{\rm(diag)}}+\delta \Gamma_a^{i{\rm(off-diag)}}
 = \frac{1}{2\bar p}\epsilon_a^{ij}\D_j
\left(\varepsilon_1+ \Delta \varepsilon_2\right).
\end{equation}
Note that this expression is diffeomorphism invariant to linear
order. Remarkably, the (gradient of the) term in the parenthesis can
be expressed as the divergence of the unsplit triad perturbation
\[
\D_j \left(\varepsilon_1+ \Delta \varepsilon_2\right) = \D_a
\delta E_j^a,
\]
which can be easily checked by inspection. Thus, for scalar mode the
linearized spin connection can be expressed as
\begin{equation}\label{Spin_Con_Lin}
\delta \Gamma_a^i =\frac{1}{2\bar p}\epsilon_a^{ij}\D_b \delta
E_j^b\,.
\end{equation}
The second order part of the gravitational Hamiltonian constraint also
contains a term quadratic in $\Gamma_i^a$. However, as such a term is
necessarily multiplied with a background quantity, the term becomes
proportional to the trace of the spin connection, $\delta_i^a
\Gamma_a^i$. For the scalar perturbation, the latter can be shown to
vanish up to at least third order using similar symmetry arguments.

In the spin connection part of the Hamiltonian, the first order
term is contributed solely by the `derivative term'
\[
\left[2 \epsilon_i^{~jk}\frac{E_j^c E_k^d }{\sqrt{|\det E|}} \D_c
\delta \Gamma_d^i\right]^{(1)}=\frac{2}{\sqrt{\bar p}} \D^i\D_a
\delta E^a_i,
\]
whereas the second order part comes from both the derivative and
the quadratic terms:
\[
\left[2 \epsilon_i^{~jk}\frac{E_j^c E_k^d }{\sqrt{|\det E|}} \D_c
\delta \Gamma_d^i\right]^{(2)}=\frac{1}{\bar p^{3/2}}
\delta^{ij}\delta E_i^a \D_a\D_b \delta E^b_j,
\]
and
\[
\left[\epsilon_i^{~jk}\frac{E_j^c E_k^d }{\sqrt{|\det
E|}}\epsilon^i_{~mn}\delta \Gamma^m_c \delta
\Gamma_d^n\right]^{(2)}=\frac{1}{2\bar p^{3/2}}
\delta^{ij}\D_a\delta E_i^a \D_b \delta E^b_j.
\]
Combining the last two terms, we obtain (up to a total divergence)
\begin{equation}
\left[\epsilon_i^{~jk}\frac{E_j^c E_k^d }{\sqrt{|\det
E|}}\left(2\D_c \delta \Gamma_d^i+\epsilon^i_{~mn}\delta
\Gamma^m_c \delta \Gamma_d^n\right)\right]^{(2)}=-\frac{1}{2\bar
p^{3/2}} \delta^{ij}\D_a\delta E_i^a \D_b \delta E^b_j.
\end{equation}
Expanding also the extrinsic curvature term, we thus arrive at the
expression for the gravitational Hamiltonian density ${\cal H}={\cal
  H}^{(0)}+ {\cal H}^{(1)}+{\cal H}^{(2)}$ with
\begin{eqnarray} \label{HamConstH0}
{\mathcal H}^{(0)} &=& -6\bar{k}^2\sqrt{\bar p}~,
%
\nonumber\\ \label{HamConstH1} {\mathcal H}^{(1)} &=&
-4 \bar{k}\sqrt{\bar{p}} \delta^c_j\delta K_c^j
-\frac{\bar{k}^2}{\sqrt{\bar{p}}} \delta_c^j\delta E^c_j
+\frac{2}{\sqrt{\bar{p}}}
\partial_c\partial^j\delta E^c_j  ~,
\nonumber\\ \label{HamConstH2} {\mathcal H}^{(2)} &=&
\sqrt{\bar{p}} \delta K_c^j\delta K_d^k\delta^c_k\delta^d_j -
\sqrt{\bar{p}} (\delta K_c^j\delta^c_j)^2
-\frac{2\bar{k}}{\sqrt{\bar{p}}} \delta E^c_j\delta K_c^j
\nonumber\\
&& \quad -\frac{\bar{k}^2}{2\bar{p}^{3/2}} \delta E^c_j\delta
E^d_k\delta_c^k\delta_d^j
+\frac{\bar{k}^2}{4\bar{p}^{3/2}}(\delta E^c_j\delta_c^j)^2
-\frac{\delta^{jk} }{2\bar{p}^{3/2}}(\partial_c\delta E^c_j)
(\partial_d\delta E^d_k)  ~.
\end{eqnarray}

Likewise, the perturbed
diffeomorphism constraint including up to quadratic order in
perturbations is
\begin{equation} \label{PertDiffConst}
D_{\rm grav}[N^a] = \frac{1}{8\pi G}\int_{\Sigma}\mathrm{d}^3x\delta N^c
\left[\bar{p}\partial_c(\delta^d_k \delta K^k_d)
-\bar{p}(\partial_k\delta K^k_c)- \bar{k} \delta_c^k(
\partial_d \delta E^d_k)\right] ~.
\end{equation}

We now consider the contribution from the scalar matter sector.
The classical Hamiltonian is given by
\begin{equation}
H_{\rm matter}[N]=\int_{\Sigma}\mathrm{d}^3x N \left({{\mathcal
H}_\pi+{\mathcal H}_\nabla+{\mathcal H}_\varphi}\right),
\end{equation}
where the kinetic, gradient and potential terms are defined in
(\ref{Hm_def}). Again, we have a perturbation expansion with
\begin{equation} \label{SFHamConstH0}
{\mathcal H}_\pi^{(0)} =
\frac{\bar{\pi}_{\bar\varphi}^2}{2\bar{p}^{3/2}} \quad,\quad
{\mathcal H}_\nabla^{(0)} = 0 \quad,\quad {\mathcal H}_\varphi^{(0)} =
\bar{p}^{3/2} V(\bar{\varphi}) ~,
\end{equation}
\begin{equation} \label{SFHamConstH1}
{\mathcal H}_\pi^{(1)} = \frac{\bar{\pi}
\delta{\pi}}{\bar{p}^{3/2}} -\frac{\bar{\pi}^2}{2\bar{p}^{3/2}}
\frac{\delta_c^j \delta E^c_j}{2\bar{p}} , \quad {\mathcal
H}_\nabla^{(1)} = 0 ,\quad {\mathcal H}_\varphi^{(1)} =
\bar{p}^{3/2}\left( V_{,\varphi}(\bar{\varphi}) \delta\varphi
+V(\bar{\varphi}) \frac{\delta_c^j \delta
E^c_j}{2\bar{p}}\right)
\end{equation}
and
\begin{eqnarray} \label{SFHamConstH2}
{\mathcal H}_\pi^{(2)} &=&
\frac{1}{2}\frac{{\delta{\pi}}^2}{\bar{p}^{3/2}} -\frac{\bar{\pi}
\delta{\pi}}{\bar{p}^{3/2}} \frac{\delta_c^j \delta
E^c_j}{2\bar{p}} +\frac{1}{2}\frac{\bar{\pi}^2}{\bar{p}^{3/2}}
\left( \frac{(\delta_c^j \delta E^c_j)^2}{8\bar{p}^2}
+\frac{\delta_c^k\delta_d^j\delta E^c_j\delta E^d_k}{4\bar{p}^2}
\right) ~, \nonumber\\
{\mathcal H}_\nabla^{(2)} &=& \frac{1}{2}\sqrt{\bar{p}}\delta^{ab}
\partial_a\delta \varphi \partial_b\delta \varphi \nonumber\\
{\mathcal H}_\varphi^{(2)} &=& \frac{1}{2}\bar{p}^{3/2}
V_{,\varphi\varphi}(\bar{\varphi}) {\delta\varphi}^2
+\bar{p}^{3/2} V_{,\varphi}(\bar{\varphi}) \delta\varphi
\frac{\delta_c^j \delta E^c_j}{2\bar{p}}
\nonumber\\
&&\quad + \bar{p}^{3/2} V(\bar{\varphi})\left( \frac{(\delta_c^j
\delta E^c_j)^2}{8\bar{p}^2} -\frac{\delta_c^k\delta_d^j\delta
E^c_j\delta E^d_k}{4\bar{p}^2} \right) ~,
\end{eqnarray}
The perturbed diffeomorphism constraint for the scalar matter field is
\begin{equation} \label{SFPertDiffConst}
D_{\rm matter}[N^a] = \int_{\Sigma}\mathrm{d}^3x\delta N^c
\bar{\pi}\partial_c\delta \varphi ~.
\end{equation}
In the following two subsections we explicitly compute the Poisson
brackets between the perturbed classical constraints and show that
their algebra is closed. At the same time, we will be including
primary correction functions and see how the algebra changes.

\subsection{Poisson bracket between Hamiltonian
and diffeomorphism constraints}

We begin by considering the gravitational sector of the classical
constraint algebra. We will see later that for computational purposes
it is convenient to split the classical perturbed gravitational
Hamiltonian as
\begin{equation} \label{ClassPertHamConst0}
H_{\rm grav}[N] =  \frac{1}{16\pi G}\int_{\Sigma}\mathrm{d}^3x N
{\mathcal H} = H_{\rm grav}[\delta N]+ H_{\rm grav}[\bar{N}] ~.
\end{equation}
Here $H_{\rm grav}[\delta N]$ includes only the perturbed component of
the lapse function whereas $H_{\rm grav}[\bar{N}]$ involves only the
background lapse. Explicit expressions for each part of the perturbed
Hamiltonian constraint are
\begin{equation}
\label{ClassPertHamConst}
H_{\rm grav}[\bar{N}] = \frac{1}{16\pi G}\int
\mathrm{d}^3x\bar{N}\left[ {\mathcal H}^{(0)}
+ {\mathcal H}^{(2)}\right] \quad,\quad
H_{\rm grav}[\delta N] =  \frac{1}{16\pi G}\int \mathrm{d}^3x
 \delta N{\mathcal H}^{(1)} \,,
\end{equation}
where perturbed Hamiltonian densities are given in equations
(\ref{HamConstH0}). We consider now the Poisson bracket between the
gravitational Hamiltonian $H_{\rm grav}[N]$ in
(\ref{ClassPertHamConst}) and the gravitational diffeomorphism
constraint $D_{\rm grav}[N^a]$ in (\ref{PertDiffConst}):
\begin{equation} \label{HDClassical}
\{H_{\rm grav}[N], D_{\rm
    grav}[N^a]\} = - H_{\rm grav}[\delta N^a\partial_a \delta N] ~.
\end{equation}
This Poisson bracket between classical perturbed constraints
(\ref{HDClassical}) is very similar to its counterpart between
the full classical constraints \cite{ALRev}, also computed in App.~\ref{a:Unpert}. This
demonstrates the consistency of perturbed constraint expressions
and elementary Poisson brackets between background and perturbed
basic variables.

As in the gravitational sector, the classical perturbed Hamiltonian
for the scalar matter field including up to quadratic terms in
perturbations can be expressed as $H_{\rm matter}[N] = H_{\rm
matter}[\delta N]+ H_{\rm matter}[\bar{N}]$ where
\begin{eqnarray} \label{SFClassPertHamConst}
H_{\rm matter}[\bar{N}] &:=& \int\mathrm{d}^3x \bar{N}\left[\left(
{\mathcal H}_\pi^{(0)}+{\mathcal H}_\varphi^{(0)}\right) +
\left({\mathcal H}_\pi^{(2)}+{ \mathcal H}_\nabla^{(2)}+
{\mathcal H}_\varphi^{(2)}\right) \right] ~, \nonumber\\
H_{\rm matter}[\delta N] &:=& \int\mathrm{d}^3x \delta N \left[
{\mathcal H}_\pi^{(1)}+{\mathcal H}_\varphi^{(1)} \right] ~.
\end{eqnarray}
Perturbed Hamiltonian densities for scalar matter are given in
equations (\ref{SFHamConstH0}), (\ref{SFHamConstH1}) and
(\ref{SFHamConstH2}). The Poisson bracket between the matter
Hamiltonian constraint and the total diffeomorphism constraint can be
computed as
\begin{equation} \label{SFHDClassical}
\{H_{\rm matter}[N], D_{\rm grav}[N^a]+D_{\rm matter}[N^a]\} = -
H_{\rm matter}[\delta N^a\partial_a \delta
N] ~.
\end{equation}
Combining gravitational sector (\ref{HDClassical}) and matter sector
(\ref{SFHDClassical}) contributions, we can evaluate the Poisson
bracket between the total Hamiltonian and diffeomorphism constraints
as
\begin{eqnarray} \label{TotalHDClassical}
\{H[N], D[N^a]\} &=& \{H_{\rm grav}[N], D_{\rm grav}[N^a]\} +
\{H_{\rm matter}[N],
D_{\rm grav}[N^a]+D_{\rm matter}[N^a]\}\nonumber\\
&=& - H_{\rm grav}[\delta N^a\partial_a \delta N] - H_{\rm matter}[\delta
N^a\partial_a \delta N] ~=~ - H[\delta N^a\partial_a \delta N] ~.
\end{eqnarray}
Clearly, perturbed expressions of total constraints along with
elementary Poisson brackets between background and perturbed basic
variables satisfy the same Poisson brackets as the full expressions.

We now analyze the situation for primary corrected constraints.  As in
the classical situation, we split the primary quantum corrected
gravitational Hamiltonian constraint as
\begin{eqnarray} \label{QuantPertHamConst} H_{\rm grav}^P[N] &=&
  \frac{1}{16\pi G}\int \mathrm{d}^3x N \alpha( E^a_i) {\mathcal H} =
  H_{\rm grav}^P[\delta N]+ H_{\rm grav}^P[\bar{N}] ~.
\end{eqnarray}
The part of Hamiltonian constraint containing only the perturbed lapse
function $H_{\rm grav}^P[\delta N]$ and the part of Hamiltonian
constraint containing only background lapse $H_{\rm grav}^P[\bar{N}]$
are defined as
\begin{eqnarray} \label{QuantPertHamConstExpr}
H_{\rm grav}^P[\bar{N}] &:=& \frac{1}{16\pi G} \int\mathrm{d}^3x
\bar{N}\left( \bar{\alpha}{\mathcal H}^{(0)}
+ \alpha^{(2)}{\mathcal H}^{(0)}
+ \alpha^{(1)}{\mathcal H}^{(1)}
+ \bar{\alpha}{\mathcal H}^{(2)}\right) \nonumber\\
H_{\rm grav}^P[\delta N] &:=& \frac{1}{16\pi G}\int\mathrm{d}^3x
 \delta N \left(\bar{\alpha}{\mathcal H}^{(1)}
+\alpha^{(1)} {\mathcal H}^{(0)}\right) ~,
\end{eqnarray}
where $\bar\alpha \equiv \alpha^{(0)}$. The Poisson bracket between
the primary quantum corrected Hamiltonian constraint
(\ref{QuantPertHamConst}) and the diffeomorphism constraint
(\ref{PertDiffConst}) can be computed as
\begin{eqnarray}\label{HDQuant}
\{H_{\rm grav}^P[N], D_{\rm grav}[N^a]\} = -
H_{\rm grav}^P[\delta N^a\partial_a\delta N]
+ {\mathcal A}_{\rm grav}^{HD}
\end{eqnarray}
where
\begin{eqnarray} \label{HDAnomaly}
{\mathcal A}_{\rm grav}^{HD} &=& -\frac{1}{16\pi G}\int\mathrm{d}^3x
(\partial_c\delta N^j)\bar{p}
\left[ \delta{N}{\mathcal H}^{(0)}
 \frac{\partial\alpha^{(1)}}{\partial(\delta E^a_i)}
(\delta^a_j \delta^c_i - \delta^c_j \delta^a_i)
+ \bar{N}\left\{ {\mathcal H}^{(0)} \left(
 -\frac{\alpha^{(1)}}{\bar{p}}\delta^c_j
\right.\right.\right. \nonumber\\ && \left.\left.\left. \quad\quad
+\frac{1}{3}\frac{\partial\bar{\alpha}}{\partial\bar{p}}
\frac{\delta E^c_j}{\bar{p}}
+ \frac{\partial\alpha^{(2)}}{\partial(\delta E^a_i)}
(\delta^a_j \delta^c_i - \delta^c_j \delta^a_i)\right)
+ {\mathcal H}^{(1)}
\frac{\partial\alpha^{(1)}}{\partial(\delta E^a_i)}
(\delta^a_j \delta^c_i - \delta^c_j \delta^a_i)
\right\} \right]
\end{eqnarray}
would appear as an anomaly if primary corrected constraints were used
as quantum constraints: the Poisson bracket (\ref{HDQuant}) between
the quantum corrected Hamiltonian constraint and diffeomorphism
constraint has additional terms which cannot be expressed completely
in terms of the gravitational constraints for any lapse function or
shift vector. Thus, these terms in the constraint algebra are
potentially anomalous.

Next we explore this issue for the quantum corrected scalar matter sector.
Similarly to the classical Hamiltonian constraint, the quantum
corrected matter Hamiltonian can be split as
\begin{eqnarray} \label{SFQuantPertHamConst}
H_{\rm matter}^P[N] &=& \int\mathrm{d}^3x N \left[ \nu(E^a_i) {\mathcal
H}_\pi +\sigma(E^a_i){\mathcal H}_\nabla +{\mathcal H}_\varphi
\right] =: H_{\rm matter}^P[\delta N]+ H_{\rm matter}^P[\bar{N}] ~.
\end{eqnarray}
The two parts $H_{\rm matter}^P[\delta N]$ and $H_{\rm
matter}^P[\bar{N}]$ of the matter Hamiltonian are defined as
\begin{eqnarray} \label{SFQuantPertHamConst1}
H_{\rm matter}^P[\bar{N}] &:=& \int\mathrm{d}^3x \bar{N} \left[
\left(\bar{\nu}{\mathcal H}_\pi^{(0)}+{\mathcal H}_\varphi^{(0)}\right)
+\left( \nu^{(2)}{\mathcal H}_\pi^{(0)}
+\nu^{(1)}{\mathcal H}_\pi^{(1)}
+\bar{\nu} {\mathcal H}_\pi^{(2)}+
 \bar{\sigma}{\mathcal H}_\nabla^{(2)}+
{\mathcal H}_\varphi^{(2)}\right) \right] \nonumber\\
H_{\rm matter}^P[\delta{N}] &:=& \int\mathrm{d}^3x \delta N
\left[\nu^{(1)}{\mathcal H}_\pi^{(0)}
+ \bar{\nu}{\mathcal H}_\pi^{(1)}
+ {\mathcal H}_\varphi^{(1)} \right] ~,
\end{eqnarray}
where $\bar{\nu}\equiv\nu^{(0)}$. Here $\nu^{(0)}$, $\nu^{(1)}$, and
$\nu^{(2)}$ denote zeroth, first and second order terms in
perturbations of the quantum correction function $\nu$. The Poisson
bracket between the quantum corrected scalar matter Hamiltonian
(\ref{SFQuantPertHamConst}) and the total diffeomorphism constraint
$D_{\rm grav}[N^a]+D_{\rm matter}[N^a]$ can be computed as
\begin{equation}\label{SFHDQuant}
\{H_{\rm matter}^P[N], D_{\rm grav}[N^a]+D_{\rm matter}[N^a]\} = -
H_{\rm matter}^P[\delta
N^a\partial_a\delta N] + {\mathcal A}_{\rm matter}^{HD}\,.
\end{equation}
As in the gravitational sector, there are additional terms
present also in the matter sector Poisson bracket which are
\begin{eqnarray} \label{SFHDAnomaly}
{\mathcal A}_{\rm matter}^{HD} &=& -\int\mathrm{d}^3x(\partial_c\delta N^j)
\bar{p}\left[ \delta{N} {\mathcal H}_\pi^{(0)}
\frac{\partial\nu^{(1)}}{\partial(\delta E^a_i)}
(\delta^a_j \delta^c_i - \delta^c_j \delta^a_i)
+ \bar{N}\left\{ {\mathcal H}_\pi^{(0)}
\left( -\frac{\nu^{(1)}}{\bar{p}}\delta^c_j
\right.\right.\right. \nonumber\\ && \left.\left.\left.
+ \frac{1}{3}\frac{\partial\bar{\nu}}{\partial\bar{p}}
\frac{\delta E^c_j}{\bar{p}}
+ \frac{\partial \nu^{(2)}}{\partial(\delta E^a_i)}
(\delta^a_j \delta^c_i - \delta^c_j \delta^a_i)\right)
+ {\mathcal H}_\pi^{(1)}
\frac{\partial\nu^{(1)}}{\partial(\delta E^a_i)}
(\delta^a_j \delta^c_i - \delta^c_j \delta^a_i)
\right\}
\right] ~.
\end{eqnarray}
Matter sector anomaly terms are similar to anomaly terms in the
gravitational sector, but there are important differences. In
particular, matter anomaly terms involve only the kinetic sector of
the matter Hamiltonian density $\mathcal H_{\pi}$.  In contrast,
gravitational anomaly terms contain the total gravitational
Hamiltonian density $\mathcal H$. Thus, one cannot even hope to
combine all anomaly terms to form the total Hamiltonian constraint for
specific correction functions. Moreover, cancellation would require a
lapse depending on correction functions. The requirement of an
anomaly-free constraint algebra then demands that gravitational sector
and matter sector anomaly terms must vanish separately. Combining
contributions from the gravitational sector (\ref{HDQuant}) and the
matter sector (\ref{SFHDQuant}), we can write the Poisson bracket
between the quantum corrected total Hamiltonian constraint and the
total diffeomorphism constraint as
\begin{eqnarray} \label{TotalHDQuant}
\{H^P[N], D[N^a]\} = - H^P[\delta N^a\partial_a \delta N]
+{\mathcal A}_{\rm grav}^{HD} + {\mathcal A}_{\rm matter}^{HD}~.
\end{eqnarray}

\subsection{Poisson bracket between two Hamiltonian constraints}

We now consider the Poisson bracket between two Hamiltonian
constraints smeared with different lapse functions.  It can be split
into three components as follows
\begin{eqnarray} \label{HHFullClass}
\{H[N_1], H[N_2]\} &=&  \{H_{\rm grav}[N_1], H_{\rm grav}[N_2]\} +
\{H_{\rm matter}[N_1], H_{\rm matter}[N_2]\} \nonumber\\ && ~+ \left[
\{H_{\rm matter}[N_1], H_{\rm grav}[N_2]\} -(N_1\leftrightarrow N_2)\right] ~.
\end{eqnarray}
Using the perturbed expression of the classical gravitational
Hamiltonian (\ref{ClassPertHamConst}) we compute the Poisson
bracket between gravitational Hamiltonian constraints as
\begin{eqnarray} \label{HHClassical}
\{H_{\rm grav}[N_1], H_{\rm grav}[N_2]\} &=&
\{H_{\rm grav}[\delta N_1],H_{\rm grav}[\bar{N}]\}+
\{H_{\rm grav}[\bar{N}],H_{\rm grav}[\delta N_2]\}\\
&=&\{H_{\rm grav}[\delta N_1-\delta N_2], H_{\rm grav}[\bar{N}]\}
=  D_{\rm grav}\left[\frac{\bar{N}}{\bar{p}}
\partial^a (\delta N_2-\delta N_1)\right] ~,\nonumber
\end{eqnarray}
where we have used the property that
$\{H_{\rm grav}[\delta N_1], H_{\rm grav}[\delta N_2]\} = 0$.
Similarly, using the perturbed expression of the classical scalar
matter Hamiltonian (\ref{SFClassPertHamConst}), we compute the
pure matter sector contribution as
\begin{equation} \label{SFHHClassical}
\{H_{\rm matter}[N_1], H_{\rm matter}[N_2]\} =
D_{\rm matter}\left[\frac{\bar{N}}{\bar{p}}
\partial^a (\delta N_2-\delta N_1)\right] ~.
\end{equation}
It is easy to show that the net contribution from the Poisson bracket
between gravitational Hamiltonian and matter Hamiltonian parts in the
constraint vanishes. In particular,
\begin{equation} \label{SFGravHHClassical}
\{H_{\rm matter}[N_1], H_{\rm grav}[N_2]\} - (N_1\leftrightarrow N_2) = 0 ~.
\end{equation}
Combining equations (\ref{HHClassical}), (\ref{SFHHClassical}) and
(\ref{SFGravHHClassical}) we evaluate the Poisson bracket between
total Hamiltonian constraints as
\begin{equation} \label{TotalHHClassical}
\{H[N_1], H[N_2]\} =  D\left[\frac{\bar{N}}{\bar{p}}
\partial^a (\delta N_2-\delta N_1)\right] ~.
\end{equation}
Thus, the perturbed expression of the classical Hamiltonian constraint
indeed satisfies the same Poisson bracket with itself as its
unperturbed expression.

Now using the perturbed expression of the primary quantum corrected
Hamiltonian (\ref{QuantPertHamConst}) we compute the Poisson
bracket between quantum corrected gravitational Hamiltonian
constraints as
\begin{equation} \label{HHQuant}
\{H_{\rm grav}^P[N_1], H_{\rm grav}^P[N_2]\} =
\{H_{\rm grav}[\delta N_1-\delta N_2], H_{\rm grav}[\bar{N}]\}
= D_{\rm grav}\left[ \bar{\alpha}^2
\frac{\bar{N} }{\bar{p}}\partial^a (\delta N_2-\delta N_1)\right]
+ {\mathcal A}_{\rm grav}^{HH}
\end{equation}
with
\begin{eqnarray} \label{HHAnomaly}
{\mathcal A}_{\rm grav}^{HH} = \frac{1}{8\pi G}\int\mathrm{d}^3x \bar{N}
(\delta N_1-\delta N_2)\left[
(\bar{\alpha} \bar{k}^2 \bar{p} \delta K_c^j)
\left\{-2\frac{\partial\bar{\alpha}}{\partial\bar{p}}\delta^c_j
+\frac{\partial \alpha^{(1)}}{\partial (\delta E^a_i)}
(\delta^c_j\delta^a_i + 3 \delta^c_i\delta^a_j) \right\}
\right. \nonumber\\ \left.
+ (2\bar{\alpha}\bar{k} \partial_c\partial^j\delta E^c_j)
\left\{ \frac{\partial\bar{\alpha}}{\partial\bar{p}} -
\frac{\partial \alpha^{(1)}}{\partial (\delta E^a_i)}\delta^a_i
\right\}
+ (6\bar{\alpha}\bar{k}^3 \bar{p}) \left\{
\frac{\partial \alpha^{(2)}}{\partial (\delta E^a_i)}\delta^a_i
-\frac{\partial \alpha^{(1)}}{\partial \bar{p}}\right\}
\right. \nonumber\\ \left.
+ (6\alpha^{(1)}\bar{k}^3 \bar{p})\left\{
\frac{\partial \bar{\alpha}}{\partial \bar{p}}
-\frac{\partial \alpha^{(1)}}
{\partial (\delta E^a_i)}\delta^a_i \right\}
+(\bar{\alpha}\bar{k}^3\delta E^c_j)\frac{\partial \alpha^{(1)}}
{\partial (\delta E^a_i)} \left\{\delta_c^j\delta^a_i
-3\delta_c^a\delta_i^j\right\}
\right] ~.
\end{eqnarray}
Similarly, using equation (\ref{SFQuantPertHamConst}) we compute
the Poisson bracket between primary quantum corrected scalar matter
Hamiltonians as
\begin{equation} \label{SFHHQuant}
\{H_{\rm matter}^P[N_1], H_{\rm matter}^P[N_2]\} =
D_{\rm matter}\left[\bar{\nu}\bar{\sigma}
\frac{\bar{N}}{\bar{p}}
\partial^a (\delta N_2-\delta N_1)\right] ~.
\end{equation}
Eq.~(\ref{SFHHQuant}) is analogous to its classical counterpart except
that the new shift vector for the resulting diffeomorphism constraint
now contains quantum correction functions $\bar{\nu}$ and
$\bar{\sigma}$.  Net contributions from the Poisson bracket between
quantum corrected gravitational Hamiltonian and matter Hamiltonian
constraints are
\begin{equation} \label{SFGravHHQuant}
\{H_{\rm matter}^P[N_1], H_{\rm grav}^P[N_2]\} - (N_1\leftrightarrow N_2)
= {\mathcal A}_m^{HH}
\end{equation}
with
\begin{eqnarray} \label{SFHHAnomaly}
{\mathcal A}_{\rm matter}^{HH} = \int\mathrm{d}^3x \bar{N}
(\delta N_1-\delta N_2)\left[
\frac{\bar{\pi}^2}{2\bar{p}^{3/2}}(\sqrt{\bar{p}} \delta K_c^j)
\left\{-\frac{2}{3}\frac{\partial\bar{\nu}}{\partial\bar{p}}\delta^c_j
+\frac{\partial \nu^{(1)}}{\partial (\delta E^a_i)}
(\delta^c_j\delta^a_i - \delta^c_i\delta^a_j) \right\}
\right. \nonumber\\ \left.
+\frac{\bar{\pi}\delta\pi}{\bar{p}^{3/2}}(2\bar{k}\sqrt{\bar{p}})
\left\{ \frac{\partial\bar{\nu}}{\partial\bar{p}} -
\frac{\partial\nu^{(1)}}{\partial(\delta E^a_i)}\delta^a_i
\right\}
-\frac{\bar{\pi}^2}{2\bar{p}^{3/2}}(2\bar{k}\sqrt{\bar{p}})
\left\{ \frac{\partial\nu^{(2)}}{\partial (\delta E^a_i)}\delta^a_i
-\frac{\partial \nu^{(1)}}{\partial \bar{p}}\right\}
\right. \nonumber\\ \left.
+\frac{\bar{\pi}^2}{2\bar{p}^{3/2}}
(\frac{\bar{k}}{\sqrt{\bar{p}}} \delta E^c_j)
\left\{-\frac{4}{3}\frac{\partial\bar{\nu}}{\partial\bar{p}}\delta_c^j
+\frac{\partial \nu^{(1)}}{\partial (\delta E^a_i)}
(\delta_c^j\delta^a_i + \delta_c^a\delta_i^j) \right\}
\right]\,.
\end{eqnarray}
Using equations (\ref{HHQuant}), (\ref{SFHHQuant}) and
(\ref{SFGravHHQuant}) we can combine contributions from the
gravitational and matter sectors to express the Poisson bracket
between primary quantum corrected total Hamiltonians:
\begin{eqnarray} \label{TotalHHQuant}
\{H^P[N_1],H^P[N_2]\} &=& D\left[\bar{\alpha}^2\frac{\bar{N}}{\bar{p}}
\partial^a (\delta N_2-\delta N_1)\right]
\nonumber\\
&+& D_{\rm matter}\left[(\bar{\nu}\bar{\sigma}-\bar{\alpha}^2)
\frac{\bar{N}}{\bar{p}}\partial^a(\delta N_2-\delta N_1)\right]
+{\mathcal A}_{\rm grav}^{HH}+ {\mathcal A}_{\rm matter}^{HH} ~.
\end{eqnarray}

\subsection{Conditions for an anomaly-free constraint algebra}

In contrast to the classical situation, we have seen that primary
quantum corrected constraints fail to form a first class constraint
algebra for arbitrary correction functions. To interpret this
properly, we recall that quantum correction functions that we have
used as a guideline in the Hamiltonian constraint are not completely
known. In particular, one can compute only zeroth order terms using
homogeneous models \cite{InvScale,Ambig,ICGC} (as well as some
partially gauge-fixed inhomogeneous cases using lattice states
\cite{QuantCorrPert}).  Linear and quadratic terms in perturbations of
quantum correction functions can in principle be computed using the
machinery of the full theory \cite{BoundFull}. However such
computations are not yet available. In this section we will analyze
whether there are conditions on quantum correction functions that we
must impose based solely on the requirement of an anomaly-free
constraint algebra.

First, we note that the explicit appearance of the matter
diffeomorphism constraint in equation (\ref{TotalHHQuant}), drops out
if the quantum correction functions satisfy
\begin{equation} \label{ANSRelation}
\bar{\alpha}^2  =  \bar{\nu} \bar{\sigma} ~.
\end{equation}
This requirement may be seen as a consistency relation between
gravitational and matter correction functions, which has important
physical implications: For instance, it ensures that gravitational
waves and massless scalar fields propagate with the same group velocity
given by the physical speed of light \cite{tensor}.

Furthermore, in order for the constraint algebra to be closed,
${\mathcal A}_{\rm grav}^{HD}$ should vanish irrespective of the
choice of lapse function.  In other words, anomaly terms involving
background lapse $\bar N$ and perturbed lapse $\delta N$ must vanish
independently.  This requirement leads to
\begin{equation} \label{HDAnomalyFreeCond}
\alpha^{(1)} = 0 \quad\mbox{ and }\quad
\frac{1}{3}\frac{\partial\bar{\alpha}}{\partial\bar{p}}
\frac{\delta E^c_j}{\bar{p}}
+ \frac{\partial\alpha^{(2)}}{\partial(\delta E^a_i)}
(\delta^a_j \delta^c_i - \delta^c_j \delta^a_i) = 0 ~.
\end{equation}
On the other hand, from equation (\ref{HHAnomaly}) the conditions
\begin{equation} \label{HHAnomalyFreeCond}
\frac{1}{3}\frac{\partial\bar{\alpha}}{\partial\bar{p}}
\delta^a_i =\frac{\partial\alpha^{(1)}}{\partial(\delta E^a_i)}
\quad\mbox{ and }\quad \frac{\partial\alpha^{(1)}}{\partial\bar{p}} =
\frac{\partial\alpha^{(2)}}{\partial(\delta E^a_i)}\delta^a_i
\end{equation}
are required to make ${\mathcal A}_{\rm grav}^{HH}$ vanish. Clearly,
the requirement of an anomaly-free constraint algebra imposes
restrictions on the first and second order terms of the quantum
correction functions.  However, it is evident that equations
(\ref{HDAnomalyFreeCond}) and (\ref{HHAnomalyFreeCond}) are
over-complete for the unknown functions $\alpha^{(1)}$ and
$\alpha^{(2)}$. Importantly, these conditions are incompatible with
each other for non-trivial primary corrections and admit only trivial
solutions of $\bar{\alpha}={\rm constant}$, $\alpha^{(1)}=0$ and
$\alpha^{(2)}=0$ which is just the classical situation without quantum
corrections. The situation for the scalar matter sector is very
similar. Using equation (\ref{SFHDAnomaly}) it is easy to see that
${\mathcal A}_{\rm grav}^{HD}=0$ requires
\begin{equation} \label{SFHDAnomalyFreeCond}
\nu^{(1)} = 0 \quad\mbox{ and }\quad
\frac{1}{3}\frac{\partial\bar{\nu}}{\partial\bar{p}}
\frac{\delta E^c_j}{\bar{p}}
+ \frac{\partial\nu^{(2)}}{\partial(\delta E^a_i)}
(\delta^a_j \delta^c_i - \delta^c_j \delta^a_i) = 0 ~,
\end{equation}
while
\begin{equation} \label{SFHHAnomalyFreeCond}
\frac{1}{3}\frac{\partial\bar{\nu}}{\partial\bar{p}}
\delta^a_i =\frac{\partial\nu^{(1)}}{\partial(\delta E^a_i)}
\quad\mbox{ and }\quad \frac{\partial\nu^{(1)}}{\partial\bar{p}} =
\frac{\partial\nu^{(2)}}{\partial(\delta E^a_i)}\delta^a_i ,
\end{equation}
solves ${\mathcal A}_{\rm matter}^{HH}=0$. As in the
gravitational sector, anomaly-free requirements on matter quantum
correction functions admit only trivial solutions of $\bar{\nu}={\rm
constant}$, $\nu^{(1)}=0$ and $\nu^{(2)}=0$. There are no additional
requirements on the quantum correction function $\sigma$ as only
its background component appears in the perturbed Hamiltonian. However,
equation (\ref{ANSRelation}) then requires that also $\bar{\sigma}$
must be constant.

At this stage, we could only conclude that inverse triad corrections
leave no trace whatsoever in effective constraints of an anomaly-free
quantization. This would be extremely puzzling given the crucial role
played by the corresponding operators for well-defined fundamental
constraints of loop quantum gravity. Fortunately, what we have shown
is, in fact, a weaker statement since we assumed the primary
correction function to depend only locally on the triad in algebraic
form. What we have shown is that this pure triad dependence is
insufficient, and we will now relax this assumption by introducing
additional corrections which we call counterterms. These terms are not
directly motivated by simple expressions computed from a
constraint operator, but they will be fixed in terms of primary
correction functions by anomaly freedom. In the conclusions, we will
comment on the expectations for the presence of such terms based on a
loop quantization.

\section{Anomaly-free quantum constraints}

In the previous section, we have shown that the presence of only
primary quantum correction functions does not lead to an anomaly-free
perturbed constraint algebra. It is however possible that the chosen
form of quantum correction functions, as they have been used in all
studies so far, does not capture all possible quantum
effects. Naturally, one is then led to ask whether there are ``counter
terms'' to the chosen form of the correction functions that should be
included in quantum corrected expressions of the Hamiltonian
constraint to make the constraint algebra anomaly-free. We show here
that such expectations are indeed realized. In particular, it is
possible to arrive at a quantum corrected constraint algebra which is
anomaly-free by including specific counter terms to the primary
quantum corrected Hamiltonian constraint.

\subsection{Gravitational sector}

For a non-trivial primary quantum correction function it is not
possible for both $\mathcal{A}_{\rm grav}^{HD}$ and $\mathcal{A}_{\rm
grav}^{HH}$ to vanish simultaneously. However, it turns out that one
can perform partial anomaly cancellation in the constraint algebra
even for non-trivial quantum correction functions by relaxing some of
the conditions imposed so far to result in manageable computations.
We approach this by adding counter terms, i.e.\ further potential
quantum corrections, ensuring first that the quantum corrected
Hamiltonian constraint is covariant under spatial diffeomorphisms,
i.e.\  $\mathcal{A}_{\rm grav}^{HD}=0$. At this point, it is important
that the diffeomorphism constraint should not receive quantum
corrections of the type studied here. Counter terms thus appear only
in the Hamiltonian constraint. We have seen that condition
(\ref{HDAnomalyFreeCond}) on the quantum correction function, in
particular $\alpha^{(1)}=0$, precisely ensures this requirement and we
can simplify the primary gravitational constraint (45) to
\begin{eqnarray} \label{QuantPertHamConstExprSimp}
H_{\rm grav}^P[\bar{N}] &:=& \frac{1}{16\pi G} \int\mathrm{d}^3x
\bar{N}\left[ \bar{\alpha}{\mathcal H}^{(0)}
+ \left(\alpha^{(2)}{\mathcal H}^{(0)}
+ \bar{\alpha}{\mathcal H}^{(2)}\right)\right] \nonumber\\
H_{\rm grav}^P[\delta N] &:=& \frac{1}{16\pi G}\int\mathrm{d}^3x
\delta N \left[ \bar{\alpha}{\mathcal H}^{(1)}\right] ~.
\end{eqnarray}
Similarly, we simplify the expression of
$\mathcal{A}_{\rm grav}^{HH}$ (59), which we then refer to as
\begin{equation} \label{AnomalyHHSimp}
\mathcal{A}_{\rm grav}^P = \frac{1}{8\pi G} \int\mathrm{d}^3x \bar{N}
(\delta N_1-\delta N_2) (2\bar{\alpha}\bar{p})
\frac{\partial \bar{\alpha}} {\partial \bar{p}}
\left[ \frac{\bar{k}}{\bar{p}}(\partial_c\partial^j\delta E^c_j)
- \bar{k}^2(\delta^c_j\delta K_c^j)
+ \bar{k}^3 \frac{(\delta_c^j\delta E^c_j)}{2\bar{p}}\right] ~.
\end{equation}

As we have already seen, canceling this anomaly based solely on
primary corrections could be achieved only for the trivial case of
constant $\alpha$.  To cancel the anomaly terms (\ref{AnomalyHHSimp})
without having to require constant $\bar{\alpha}$, we need to generate
additional terms which take a form similar to primary anomaly
terms. In fact, additional corrections not considered so far can
easily arise in an effective Hamiltonian constraint. Adding additional
terms into the Hamiltonian constraint however can potentially generate
new anomaly terms in the Poisson bracket $\{H_{\rm grav}^P[N], D_{\rm
grav}[N^a]\}$. To avoid this we should ensure that the counter terms
being added to the Hamiltonian constraint commute with the
diffeomorphism constraint. Moreover, counter terms should not affect
the background dynamics as there are no anomalies in the constraint
algebra when one turns off inhomogeneity. Thus, counter terms should be
constructed using only those terms which contain inhomogeneous
perturbations.

We notice that while we describe the emergence of counter terms in a
constructive manner, all this reflects requirements which fundamental
anomaly-freedom would pose. To that end, we consider a minimal
approach for constructing counter terms. In particular, we require
that counter terms should generate only those three kinds of terms
which are already present in the anomaly expression
(\ref{AnomalyHHSimp}). With this requirement, the allowed form of
counter terms that can be included in the quantum corrected
Hamiltonian is given by $H_{\rm grav}^C[N] = H_{\rm grav}^C[\delta
N]+ H_{\rm grav}^C[\bar{N}]$ where
\begin{equation} \label{GravCounterTermDeltaN}
H_{\rm grav}^C[\delta N] = \frac{1}{16\pi G} \int\mathrm{d}^3x \delta N
\bar{\alpha}\left[-4 f(\bar{p})\bar{k}\sqrt{\bar{p}}
(\delta^c_j\delta K_c^j) - g(\bar{p})
\frac{\bar{k}^2}{\sqrt{\bar{p}}}(\delta_c^j\delta E^c_j) \right]
\end{equation}
and
\begin{equation} \label{GravCounterTermBarN}
H_{\rm grav}^C[\bar{N}] = \frac{1}{16\pi G}\int\mathrm{d}^3x \bar{N}
\bar{\alpha}\left[- h(\bar{p})\frac{\delta^{jk}}{2\bar{p}^{3/2}}
(\partial_c\delta E^c_j)(\partial_d\delta E^d_k) \right] ~.
\end{equation}
Here we have introduced three dimensionless scalar functions $f$, $g$
and $h$ which depend on the quantum correction functions and are to be
determined through anomaly cancellation conditions.  Only background
components of these functions are relevant as the counter terms are
already quadratic in perturbations.

The new terms can be interpreted as resulting from a dependence of the
primary $\alpha$ on extrinsic curvature components and spatial
derivatives of the triad as a general functional. Thus, the
introduction of counter terms relaxes some of the conditions imposed
earlier on $\alpha$. The Poisson bracket between counter terms and the
diffeomorphism constraint can be computed as
\begin{eqnarray} \label{CTDiffeoPB}
\{H_{\rm grav}^C[N],D[N^a]\} = \frac{1}{8\pi G}\int \mathrm{d}^3x
\bar{\alpha} (\partial_c \delta N^c)\delta N \left[
-(2f+g)(\bar{k}^2\sqrt{\bar{p}})\right] ~.
\end{eqnarray}
Thus, the requirement that counter terms commute with the
diffeomorphism constraint leads to the simple condition
\begin{equation} \label{CTEqn1}
2f+g=0
\end{equation}
on the coefficient functions. Out of three unknown coefficient
functions only two remain to be determined. Using equations
(\ref{GravCounterTermDeltaN}), (\ref{GravCounterTermBarN}) we compute
contributions from counter terms to the Poisson bracket between
Hamiltonian constraints
\begin{eqnarray} \label{CTContribGrav}
\mathcal{A}_{\rm grav}^C &=& \{H_{\rm grav}^C[N_1],H_{\rm grav}^P[N_2]\}
-(N_1 \leftrightarrow N_2) \nonumber\\
&=& \{H_{\rm grav}^C[\delta N_1-\delta N_2],H_{\rm grav}^P[\bar{N}]\}
+ \{H_{\rm grav}^P[\delta N_1-\delta N_2],H_{\rm grav}^C[\bar{N}]\} \nonumber\\
&=& \frac{1}{8\pi G} \int\mathrm{d}^3x \bar{N} (\delta N_1-\delta
N_2)\bar{\alpha}^2 \left[ \bar{k}^2(\delta^c_j\delta K_c^j)
\left(f-g-4\bar{p}\frac{\partial f}{\partial \bar{p}} \right)
\right.\nonumber\\ &&\quad\quad\quad\quad + \left.\bar{k}^3
\frac{(\delta_c^j\delta E^c_j)}{2\bar{p}}
\left(-f+g-2\bar{p}\frac{\partial g}{\partial \bar{p}} \right)
+\frac{\bar{k}}{\bar{p}} (\partial_c\partial^j\delta E^c_j)
\left(-h-f\right) \right] ~.
\end{eqnarray}

Since the perturbed classical constraint algebra is closed without
counter terms, the quantum counter terms must vanish in the classical
limit, i.e.\ when all primary correction function are unity.  Counter
terms can then depend on primary correction functions only through
their derivatives. Given the expanded form (\ref{AlphaHomo}) of
primary quantum correction functions used here, terms such as
$(\partial\bar\alpha / \partial \bar p)^2$ can be neglected compared
to the terms $\partial\bar\alpha / \partial \bar p$. For the same
reason the contributions from the Poisson bracket between counter
terms $\{H_{\rm grav}^C[N_1],H_{\rm grav}^C[N_2]\}$ can be neglected
compared to the contributions considered in (\ref{CTContribGrav}).

Combining all contributions from the original anomaly
(\ref{AnomalyHHSimp}) and from counter terms (\ref{CTContribGrav}) one
can express the total gravitational anomaly $\mathcal{A}_{\rm grav}
:= \mathcal{A}_{\rm grav}^{P} + \mathcal{A}_{\rm grav}^{C}$ as
\begin{equation} \label{TotalGravAnomalyHH}
\mathcal{A}_{\rm grav} = \frac{1}{8\pi G} \int\mathrm{d}^3x \bar{N}
(\delta N_1-\delta N_2)\bar{\alpha}^2 \left[ \frac{\bar{k}}{\bar{p}}
(\partial_c\partial^j\delta E^c_j) \mathcal{G}_1
+ \bar{k}^2(\delta^c_j\delta K_c^j) \mathcal{G}_2
+ \bar{k}^3 \frac{(\delta_c^j\delta E^c_j)}{2\bar{p}}
\mathcal{G}_3 \right] ~,
\end{equation}
where
\begin{eqnarray} \label{AnoCanGravEqn1}
\mathcal{G}_1 &=& -h -f+\frac{2\bar{p}}{\bar{\alpha}}\frac{\partial
\bar{\alpha}} {\partial \bar{p}} ~,\\
\label{AnoCanGravEqn2}
\mathcal{G}_2 &=& f-g-4\bar{p}\frac{\partial f}{\partial \bar{p}}
-\frac{2\bar{p}}{\bar{\alpha}}\frac{\partial \bar{\alpha}}
{\partial \bar{p}} ~, \\
\label{AnoCanGravEqn3}
\mathcal{G}_3 &=& -f+g-2\bar{p}\frac{\partial g}{\partial \bar{p}}
+\frac{2\bar{p}}{\bar{\alpha}}\frac{\partial \bar{\alpha}}
{\partial \bar{p}} ~,
\end{eqnarray}
Anomaly cancellation will require coefficients of
($\partial_c\partial^j\delta E^c_j$), ($\delta^c_j\delta K_c^j$) and
($\delta_c^j\delta E^c_j$) to vanish. This in turn implies that the
coefficient functions $f$, $g$ and $h$ should be such that they
satisfy three equations $\mathcal{G}_1=0$, $\mathcal{G}_2=0$ and
$\mathcal{G}_3=0$. On the other hand, $f$ and $g$ also need to satisfy
equation (\ref{CTEqn1}). Thus, the set of equations for $f$, $g$ and
$h$ may appear to be over-complete.  However, it is remarkable to note
that equation (\ref{CTEqn1}) along with equation
(\ref{AnoCanGravEqn2}) solves the equation (\ref{AnoCanGravEqn3})
identically. In particular using equation (\ref{CTEqn1}), it is easy
to see that
\begin{equation} \label{GravConsistencyRelation}
\mathcal{G}_3 = - \mathcal{G}_2 ~.
\end{equation}
Thus for a given quantum correction function $\alpha$, there are
unambiguous solutions for $f$, $g$ and $h$ such that the constraint
algebra is anomaly-free. In particular, for a background
quantum correction function $\alpha$ given in (\ref{AlphaHomo}),
\begin{equation}\label{FGHSolns}
 f = -\frac{g}{2} =
\frac{2}{4n_{\alpha}+3}\frac{\bar{p}}{\bar{\alpha}}
\frac{\partial \bar{\alpha}}{\partial \bar{p}}\quad,\quad h =
4\frac{2n_{\alpha}+1}{4n_{\alpha}+3}\frac{\bar{p}}{\bar{\alpha}}
\frac{\partial \bar{\alpha}} {\partial \bar{p}}\quad,
\end{equation}
solve equation (\ref{CTEqn1}) and ensure vanishing of equations (\ref{AnoCanGravEqn1}),
(\ref{AnoCanGravEqn2}) and (\ref{AnoCanGravEqn3}).

\subsection{Cosmological constant}

For obtaining anomaly freedom in the gravitational sector by adding
appropriate counter terms, it was crucial that coefficients
$\mathcal{G}_2$ and $\mathcal{G}_3$ are related through equation
(\ref{GravConsistencyRelation}). Thus, one should consider the
robustness of this relation under the inclusion of additional
classical terms. Including a non-zero cosmological constant to
the gravitational sector provides a definite test to see whether such a
relation can still be satisfied. Counter terms in the gravitational
sector would now generate additional contributions due to the presence
of the cosmological constant term. We now show that a
non-zero cosmological constant does not spoil the non-trivial
consistency condition (\ref{GravConsistencyRelation}).

Contributions
to the gravitational Hamiltonian constraint from a non-zero
cosmological constant $\Lambda$ are
\begin{equation} \label{CCFull}
H_{\Lambda}[N] = \frac{1}{8\pi G}\int\mathrm{d}^3x
N \sqrt{|\det E|} \Lambda
=: H_{\Lambda}[\delta N]+H_{\Lambda}[\bar{N}] ~.
\end{equation}
As with the matter potential term, no primary inverse triad
corrections are expected.  The perturbed expressions of $
H_{\Lambda}[\delta N]$ and $H_{\Lambda}[\bar{N}]$, including up to
quadratic terms in perturbations are given by
\begin{equation} \label{CCBarN}
H_{\Lambda}[\bar{N}] = \frac{1}{16\pi G}\int\mathrm{d}^3x
\bar{N}\left[ {\mathcal H}_{\Lambda}^{(0)}+{\mathcal
H}_{\Lambda}^{(2)}\right] \quad,\quad
H_{\Lambda}[\delta N] = \frac{1}{16\pi G} \int\mathrm{d}^3x
\delta N{\mathcal H}_{\Lambda}^{(1)}~.
\end{equation}
The explicit expressions of perturbed densities ${\mathcal
H}_{\Lambda}^{(0)}$, ${\mathcal H}_{\Lambda}^{(1)}$ and ${\mathcal
H}_{\Lambda}^{(2)}$ are
\begin{eqnarray} \label{CCDeltaN}
&&{\mathcal H}_{\Lambda}^{(0)} = 2\Lambda \bar{p}^{3/2} \quad,\quad
{\mathcal H}_{\Lambda}^{(1)} = 2\Lambda \bar{p}^{3/2}
\frac{(\delta_c^j\delta E^c_j)}{2\bar{p}} \\
&&{\mathcal H}_{\Lambda}^{(2)} = 2\Lambda\bar{p}^{3/2}
\left(\frac{(\delta_c^j \delta E^c_j)^2}{8\bar{p}^2}
-\frac{(\delta_c^k\delta_d^j\delta E^c_j\delta E^d_k)}{4\bar{p}^2}
\right) ~.
\end{eqnarray}
Contributions to the anomaly expression arising from the
Poisson bracket between counter terms and cosmological constant terms
are
\begin{eqnarray} \label{CTContribGravCC0}
\mathcal{A}_{\Lambda} &=& \{H_{\rm grav}^C[N_1],H_{\Lambda}[N_2]\}
-(N_1 \leftrightarrow N_2) \nonumber\\
&=& \frac{1}{8\pi G}
\int\mathrm{d}^3x \bar{N}(\delta N_1-\delta N_2)\bar{\alpha}
(\Lambda \bar{p}) \left[ - (\delta^c_j\delta K_c^j) f -
\bar{k} \frac{(\delta_c^j\delta E^c_j)}{2\bar{p}}
(f+g) \right] ~.
\end{eqnarray}
Combining contributions from counterterms and the original anomaly in
the presence of a non-zero cosmological constant, one can evaluate the
total gravitational anomaly $\mathcal{A}_{\rm grav} :=
\mathcal{A}_{\rm grav}^{P} + \mathcal{A}_{\rm grav}^{C}+
\mathcal{A}_{\Lambda}$ as
\begin{equation} \label{TotalGravAnomalyHHCC}
\mathcal{A}_{\rm grav} = \frac{1}{8\pi G} \int\mathrm{d}^3x \bar{N}
(\delta N_1-\delta N_2)\bar{\alpha}^2 \left[ \frac{\bar{k}}{\bar{p}}
(\partial_c\partial^j\delta E^c_j) \mathcal{G}_1^{\Lambda}
+ \bar{k}^2(\delta^c_j\delta K_c^j) \mathcal{G}_2^{\Lambda}
+ \bar{k}^3 \frac{(\delta_c^j\delta E^c_j)}{2\bar{p}}
\mathcal{G}_3^{\Lambda} \right] ~,
\end{equation}
where new coefficients in the anomaly expression are
\begin{eqnarray} \label{AnoCanGravCCEqn}
\mathcal{G}_1^{\Lambda}=\mathcal{G}_1 \quad,\quad
\mathcal{G}_2^{\Lambda} = \mathcal{G}_2 -
f\frac{\Lambda \bar{p}}{\bar{\alpha}\bar{k}^2} \quad,\quad
\mathcal{G}_3^{\Lambda} = \mathcal{G}_3 -
(f+g)\frac{\Lambda \bar{p}}{\bar{\alpha}\bar{k}^2}~.
\end{eqnarray}
(The $\Lambda$-dependence cancels upon using the background Friedmann
equation.)  Using equation (\ref{CTEqn1}), which remains unchanged, we
again note that the new coefficients satisfy the same non-trivial
consistency relation
\begin{equation} \label{GravConsistencyRelationCC}
\mathcal{G}_3^{\Lambda} = - \mathcal{G}_2^{\Lambda} ~.
\end{equation}
Thus, also in the presence of a non-zero cosmological constant there
are unambiguous solutions for $f$, $g$ and $h$ such that the
constraint algebra is anomaly-free. This demonstrates that
anomaly freedom of the quantum corrected constraint algebra including
appropriate counter terms is a robust feature.

\subsection{Scalar matter}

For cosmological applications, we must ensure the existence of
consistent equations in the presence of matter.  Similarly to the
gravitational sector we ensure first that the quantum corrected scalar
matter Hamiltonian is covariant under spatial diffeomorphism,
i.e.\  $\mathcal{A}_{\rm matter}^{HD}=0$. This requirement implies
that we should impose conditions (\ref{SFHDAnomalyFreeCond}) on the
quantum correction function $\nu$ and simplify the expression of the
primary quantum corrected matter Hamiltonian as
\begin{eqnarray} \label{SFQuantPertHamConstSimp}
H_{\rm matter}^P[\bar{N}] &:=& \int\mathrm{d}^3x \bar{N} \left[
\left(\bar{\nu}{\mathcal H}_\pi^{(0)}+{\mathcal
H}_\varphi^{(0)}\right)
+ \left( \nu^{(2)}{\mathcal H}_\pi^{(0)}
+\bar{\nu} {\mathcal H}_\pi^{(2)}+
+ \bar{\sigma}{\mathcal H}_\nabla^{(2)}+
{\mathcal H}_\varphi^{(2)}\right) \right] \nonumber\\
H_{\rm matter}^P[\delta{N}] &:=& \int\mathrm{d}^3x
\delta N \left[ \bar{\nu}{\mathcal H}_\pi^{(1)}
+ {\mathcal H}_\varphi^{(1)} \right] ~.
\end{eqnarray}
This also simplifies the matter anomaly term $\mathcal{A}_{\rm
matter}^{HH}$ which we then refer to as
\begin{eqnarray} \label{SFHHAnomalySimp}
{\mathcal A}_{\rm matter}^P &=&
 \int\mathrm{d}^3x \bar{N}(\delta N_1-\delta N_2)
\left[\frac{\bar{\nu}\bar{\pi}\delta\pi}{\bar{p}^{3/2}}
\frac{\bar{\alpha}\bar{k}}{\sqrt{\bar{p}}}
\left(\frac{2\bar{p}}{\bar{\nu}}\frac{\partial\bar{\nu}}{\partial\bar{
p}}\right) +
\frac{\bar{\nu}\bar{\pi}^2}{2\bar{p}^{3/2}}(\delta^c_j\delta K_c^j)
\frac{\bar{\alpha}}{\sqrt{\bar{p}}} \left(-\frac{2\bar{p}}{3\bar{\nu}}
\frac{\partial\bar{\nu}}{\partial\bar{p}}\right)
\right. \nonumber\\ && \left.
+ \frac{\bar{\nu}\bar{\pi}^2}{2\bar{p}^{3/2}}
\frac{(\delta_c^j\delta E^c_j)}{2\bar{p}}
\frac{\bar{\alpha}\bar{k}}{\sqrt{\bar{p}}}
\left(-\frac{10\bar{p}}{3\bar{\nu}}
\frac{\partial\bar{\nu}}{\partial\bar{p}}\right)
\right] ~.
\end{eqnarray}
For computational convenience we consider the construction of counter
terms for the kinetic and potential sectors separately.

\subsubsection{Kinetic sector}

To cancel anomalies in the kinetic sector of scalar matter, we start
with a general form of possible counter terms as $H_{\pi\nabla}^C[N]
:= H_{\pi}^C[\delta N]+H_{\nabla}^C[\delta N] + H_{\pi}^C[\bar{N}] +
H_{\nabla}^C[\bar{N}]$ where
\begin{equation} \label{SFGenCounterTermKSDeltaN}
H_{\pi}^C[\delta N] = \int\mathrm{d}^3x \delta N \left[
f_1(\bar{p})\frac{\bar{\nu}\bar{\pi}\delta{\pi}}{\bar{p}^{3/2}}
-f_2(\bar{p})\frac{\bar{\nu}\bar{\pi}^2}{2\bar{p}^{3/2}}
\frac{(\delta_c^j
\delta E^c_j)}{2\bar{p}} \right] \quad,\quad
H_{\nabla}^C[\delta N] = 0
\end{equation}
and
\begin{eqnarray} \label{SFGenCounterTermKSBarN}
H_{\pi}^C[\bar{N}] &=& \int\mathrm{d}^3x\bar{N}\left[
g_1(\bar{p})\frac{\bar{\nu}\delta{\pi}^2}{2\bar{p}^{3/2}}
-g_2(\bar{p})\frac{\bar{\nu}\bar{\pi} \delta\pi}{\bar{p}^{3/2}}
\frac{(\delta_c^j \delta E^c_j)}{2\bar{p}}
\right. \nonumber\\
&& ~~~~~~~~~~~~~~ +
\left. \frac{\bar{\nu}\bar{\pi}^2}{2\bar{p}^{3/2}} \left(
g_3(\bar{p})\frac{(\delta_c^j \delta E^c_j)^2}{8\bar{p}^2}
+g_4(\bar{p})\frac{(\delta_c^k\delta_d^j\delta E^c_j\delta
E^d_k)}{4\bar{p}^2}
\right)\right] ~, \nonumber\\
H_{\nabla}^c[\bar{N}] &=& \int\mathrm{d}^3x \bar{N} \left[
g_5(\bar{p})\frac{\bar{\sigma}\sqrt{\bar{p}}}{2}\delta^{ab}
\partial_a\delta \varphi \partial_b\delta \varphi \right]~.
\end{eqnarray}
The general guidelines followed for the gravitational sector led us to
introduce seven dimensionless unknown functions $f_1$, $f_2$, $g_1$,
$g_2$, $g_3$, $g_4$ and $g_5$ which should be related to primary
quantum correction functions as it will be determined through anomaly
cancellation conditions.

The Poisson bracket between counter terms
(\ref{SFGenCounterTermKSDeltaN}), (\ref{SFGenCounterTermKSBarN}) and
the total diffeomorphism constraint is
\begin{eqnarray} \label{SFCTDiffeoPBKS}
\{H_{\pi\nabla}^C[N],D[N^a]\} =
\int\mathrm{d}^3x \left[ (\partial_c \delta N^j)\bar{N}\left(g_4
\frac{\bar{\nu}\bar{\pi}^2}{2\bar{p}^{3/2}}
\frac{\delta E^c_j}{2\bar{p}}\right)
+ (\partial_c \delta N^c) \left\{
\delta N (2f_1 -f_2)\frac{\bar{\nu}\bar{\pi}^2}{2\bar{p}^{3/2}}
\right.\right. \nonumber\\ \left. \left.
+ \bar{N}(g_1-g_2)
\frac{\bar{\nu}\bar{\pi}\delta{\pi}}{\bar{p}^{3/2}}
-\bar{N}(2g_2-g_3+g_4) \frac{\bar{\nu}\bar{\pi}^2}{2\bar{p}^{3/2}}
\frac{(\delta_c^j\delta E^c_j)}{2\bar{p}} \right\}
\right] ~.
\end{eqnarray}
Requiring that counter terms commute with the diffeomorphism
constraint leads to the conditions
\begin{equation} \label{SFCTEqn1}
2f_1 = f_2 ~,~ g_1=g_2 ~,~ 2g_2=g_3 ~,~ g_4 = 0 ~.
\end{equation}
Given that $g_4$ is required to vanish, we will drop the corresponding
term from the set of counter terms in our further evaluation.  We
began with seven unknown functions in the counter terms for the
kinetic sector. Requiring that counter terms commute with the
diffeomorphism constraint imposes four conditions. This in turn allows
only three more conditions to be imposed on these functions from
anomaly cancellation in the remaining Poisson bracket of Hamiltonian
constraints.

There are some subtleties in finding anomaly cancellation conditions
for the matter sector. Inclusion of counter terms to the gravitational
sector has generated additional contributions both in the
gravitational as well as the matter sector. Thus, matter sector
counter terms need to cancel the original anomaly expression
(\ref{SFHHAnomalySimp}) but also contributions from gravitational
counter terms.  Contributions from gravitational counter terms to
anomaly expressions of the matter kinetic sector are
\begin{eqnarray} \label{GravCTContribSFKS}
\mathcal{A}_{{\rm grav}\pi\nabla}^{C} &:=& \{H_{\rm grav}^C[N_1],H_{\pi\nabla}^P[N_2]\}
-(N_1\leftrightarrow N_2) \nonumber\\
&=&\int\mathrm{d}^3x \bar{N}(\delta N_1-\delta N_2)
\left[\frac{\bar{\nu}\bar{\pi}\delta\pi}{\bar{p}^{3/2}}
\frac{\bar{\alpha}\bar{k}}{\sqrt{\bar{p}}}(3f) +
\frac{\bar{\nu}\bar{\pi}^2}{2\bar{p}^{3/2}}(\delta^c_j\delta K_c^j)
\frac{\bar{\alpha}}{\sqrt{\bar{p}}}(f)
\right. \nonumber\\ && \left.
+ \frac{\bar{\nu}\bar{\pi}^2}{2\bar{p}^{3/2}}
\frac{(\delta_c^j\delta E^c_j)}{2\bar{p}}
\frac{\bar{\alpha}\bar{k}}{\sqrt{\bar{p}}}(g-5f)\right] ~.
\end{eqnarray}
Similarly, contributions from counter terms of the matter kinetic
sector to the Poisson bracket between total Hamiltonians can be
computed as
\begin{eqnarray} \label{SFCTContribSFKS}
\mathcal{A}_{\pi\nabla}^{C} &:=&
\{H_{\pi\nabla}^C[N_1],H_{\rm grav}^P[N_2]+H_{\pi\nabla}^P[N_2]\}
-(N_1 \leftrightarrow N_2) \nonumber\\
&=&\int\mathrm{d}^3x \bar{N}(\delta N_1-\delta N_2) \left[
\frac{\bar{\nu}\bar{\pi}}{\bar{p}}\bar{\sigma}
\nabla^2\delta\varphi(f_1+g_5) +
\frac{\bar{\nu}\bar{\pi}\delta\pi}{\bar{p}^{3/2}}
\frac{\bar{\alpha}\bar{k}}{\sqrt{\bar{p}}}
\left(2\bar{p}\frac{\partial f_1}{\partial\bar{p}}-3f_1+3g_2\right)
 \right. \nonumber\\ &&
+ \left.
\frac{\bar{\nu}\bar{\pi}^2}{2\bar{p}^{3/2}}(\delta^c_j\delta K_c^j)
\frac{\bar{\alpha}}{\sqrt{\bar{p}}}(-f_2) +
\frac{\bar{\nu}\bar{\pi}^2}{2\bar{p}^{3/2}}\frac{(\delta_c^j\delta
E^c_j)}{2\bar{p}}
\frac{\bar{\alpha}\bar{k}}{\sqrt{\bar{p}}}
\left(-2\bar{p}\frac{\partial f_2}{\partial\bar{p}}+4f_2-3g_3\right)
\right] ~,
\end{eqnarray}
We then combine the original anomaly (\ref{SFHHAnomalySimp}) with
contributions (\ref{GravCTContribSFKS}) from gravitational counter
terms and contributions from matter kinetic sector counter terms
(\ref{SFCTContribSFKS}) to express the total anomaly
$\mathcal{A}_{\pi\nabla} := \mathcal{A}_{\rm matter}^P +
\mathcal{A}_{{\rm grav}\pi\nabla}^{C} + \mathcal{A}_{ \pi\nabla}^C$ in
the kinetic sector of scalar matter as
\begin{eqnarray} \label{SFTotalKinAnomaly}
\mathcal{A}_{\pi\nabla} &=& \int\mathrm{d}^3x\bar{N}
(\delta N_1-\delta N_2)
\left[\frac{\bar{\nu}\bar{\pi}\delta\pi}{\bar{p}^{3/2}}
\frac{\bar{\alpha}\bar{k}}{\sqrt{\bar{p}}}
\mathcal{B}_1 +
\frac{\bar{\nu}\bar{\pi}^2}{2\bar{p}^{3/2}}(\delta^c_j\delta K_c^j)
\frac{\bar{\alpha}}{\sqrt{\bar{p}}}
\mathcal{B}_2
\right. \nonumber\\
&& \left. + \frac{\bar{\nu}\bar{\pi}^2}{2\bar{p}^{3/2}}
\frac{(\delta_c^j\delta E^c_j)}{2\bar{p}}
\frac{\bar{\alpha}\bar{k}}{\sqrt{\bar{p}}}
\mathcal{B}_3
+
\frac{\bar{\nu}\bar{\pi}}{\bar{p}}\bar{\sigma}
\nabla^2\delta\varphi \mathcal{B}_4
\right] ~,
\end{eqnarray}
where
\begin{eqnarray} \label{EqnB1}
\mathcal{B}_1 &=& \frac{2\bar{p}}{\bar{\nu}}\frac{\partial\bar{\nu}}
{\partial\bar{p}} + 3f +
2\bar{p}\frac{\partial f_1}{\partial\bar{p}}-3f_1+3g_2
\\ \label{EqnB2}
\mathcal{B}_2 &=& -\frac{2\bar{p}}{3\bar{\nu}}
\frac{\partial\bar{\nu}}{\partial\bar{p}}+f-f_2
\\ \label{EqnB3}
\mathcal{B}_3 &=&
-\frac{10\bar{p}}{3\bar{\nu}}\frac{\partial\bar{\nu}}
{\partial\bar{p}}-5f+g-2\bar{p}\frac{\partial f_2}{\partial\bar{p}}
+4 f_2 - 3 g_3
\\ \label{EqnB4}
\mathcal{B}_4 &=& f_1+g_5 ~.
\end{eqnarray}
In the presence of a non-zero scalar matter potential, the imposition
of background and perturbed Hamiltonian constraints does not determine
matter kinetic terms in terms of gravitational terms in the Hamiltonian
constraint. In such situations anomalies in the kinetic sector of
scalar matter must vanish independently of the gravitational sector
anomaly. From, Eq.~(\ref{SFTotalKinAnomaly}) it is evident that
anomaly freedom in the kinetic sector requires four conditions to be
satisfied, i.e.\ $\mathcal{B}_1=0$, $\mathcal{B}_2=0$,
$\mathcal{B}_3=0$ and $\mathcal{B}_4=0$. However as we mentioned,
after imposing equations (\ref{SFCTEqn1}) we have only three
undetermined functions in kinetic sector counter terms.  Thus it may
appear once again that there is over-determination of counter
terms. However, similarly to the gravitational sector, coefficients of
the anomaly expression (\ref{SFTotalKinAnomaly}) satisfy a non-trivial
consistency relation. In particular, using relations (\ref{CTEqn1}),
(\ref{SFCTEqn1}), one notes that
\begin{equation} \label{RelationB123}
\mathcal{B}_3 = -2 \mathcal{B}_1 - \mathcal{B}_2 ~.
\end{equation}
Thus, cancellation of kinetic sector anomalies requires only three
equations to be satisfied by counter term coefficients. In other
words, counter terms for kinetic sector are unambiguously determined
by anomaly cancellation conditions.

\subsubsection{Potential sector}

We recall that the potential sector did not involve any primary
quantum correction function and did not contribute to the matter
sector anomaly.  However, including counter terms in the gravitational
and matter kinetic sectors leads to new anomaly terms involving the
scalar matter potential. Such new anomaly contributions from
gravitational counter terms to the Poisson bracket between Hamiltonians
can be computed as
\begin{eqnarray} \label{GravCTContribSFPot}
\mathcal{A}_{{\rm grav}\varphi}^{C} &:=& \{H_{\rm grav}^C[N_1],
H_{\varphi}^P[N_2]\}
-(N_1\leftrightarrow N_2) \nonumber\\
&=&\int\mathrm{d}^3x \bar{N}(\delta N_1-\delta N_2) \left[
\frac{\bar{\alpha}}{\sqrt{\bar{p}}} \bar{p}^{3/2} V(\bar{\varphi})
(\delta^c_j\delta K_c^j)(-f) +
\frac{\bar{\alpha}\bar{k}}{\sqrt{\bar{p}}} \bar{p}^{3/2}
V_{,\varphi}(\bar{\varphi}) \delta\varphi (-3f) \right. \nonumber\\
&& \left. ~~~~~~ + \frac{\bar{\alpha}\bar{k}}{\sqrt{\bar{p}}}
\bar{p}^{3/2} V(\bar{\varphi}) \frac{(\delta_c^j\delta
E^c_j)}{2\bar{p}}(-f-g) \right] ~.
\end{eqnarray}
Similar anomaly contributions from counter terms of the matter kinetic
sector are
\begin{eqnarray}
\label{KSCTContribSFPot}
\mathcal{A}_{\varphi\pi}^{C} &:=&
\{H_{\pi\nabla}^C[N_1],H_{\varphi}^P[N_2]\}
-(N_1\leftrightarrow N_2) \nonumber\\
&=&\int\mathrm{d}^3x \bar{N}(\delta N_1-\delta N_2) \left[
\frac{\bar{\nu}\bar{\pi}}{\bar{p}^{3/2}} \bar{p}^{3/2}
V_{,\varphi\varphi}(\bar{\varphi})\delta\varphi(-f_1) +
\frac{\bar{\nu}\delta \pi}{\bar{p}^{3/2}} \bar{p}^{3/2}
V_{,\varphi}(\bar{\varphi}) (-f_1 + g_1) \right. \nonumber\\ &&
\left. ~~~~~~ + \frac{\bar{\nu}\bar{\pi}}{\bar{p}^{3/2}}
\bar{p}^{3/2} V_{,\varphi}(\bar{\varphi}) \frac{(\delta_c^j\delta
E^c_j)}{2\bar{p}}(f_2-f_1-g_2) \right] ~.
\end{eqnarray}
Thus, counter terms of the matter kinetic sector generate a new
anomaly term involving $V_{,\varphi\varphi}(\bar{\varphi})$.
Gravitational counter terms on the other hand, do not lead to such a
term. For non-vanishing $f_1$, not all terms can cancel by combining
equations (\ref{GravCTContribSFPot}) and
(\ref{KSCTContribSFPot}). Even though we did not consider primary
quantum corrections in the potential sector, for anomaly freedom we
need to allow counter terms even here. As in the kinetic
sector, we begin with a general expression of possible counter terms
in the potential sector $H_{\varphi}^C[N] := H_{\varphi}^C[\delta N]
+H_{\varphi}^C[\bar{N}]$ where
\begin{equation} \label{SFGenCounterTermPotDeltaN}
H_{\varphi}^C[\delta N] = \int\mathrm{d}^3x \delta N\left[
f_3(\bar{p}) \bar{p}^{3/2} V_{,\varphi}(\bar{\varphi})
\delta\varphi +f_4(\bar{p})\bar{p}^{3/2} V(\bar{\varphi})
\frac{(\delta_c^j \delta E^c_j)}{2\bar{p}}\right] ~,
\end{equation}
and
\begin{eqnarray} \label{SFGenCounterTermPotBarN}
H_{\varphi}^C[\bar{N}] &=& \int\mathrm{d}^3x\bar{N} \left[
g_6(\bar{p}) \frac{1}{2}\bar{p}^{3/2}
V_{\varphi\varphi}(\bar{\varphi}) {\delta\varphi}^2 + g_7(\bar{p})
\bar{p}^{3/2} V_{,\varphi}(\bar{\varphi}) \delta\varphi
\frac{(\delta_c^j \delta E^c_j)}{2\bar{p}}
\right. \nonumber\\
&& \left. \quad + \bar{p}^{3/2} V(\bar{\varphi})\left(
g_8(\bar{p})\frac{(\delta_c^j \delta E^c_j)^2}{8\bar{p}^2}
- g_9(\bar{p}) \frac{\delta_c^k\delta_d^j\delta
E^c_j\delta E^d_k}{4\bar{p}^2} \right)\right] ~.
\end{eqnarray}
We have introduced six new unknown functions $f_3$, $f_4$, $g_6$, $g_7$,
$g_8$ and $g_9$ in the counter terms of the potential sector. To
ensure invariance of the potential sector counter terms under
diffeomorphism constraint we compute the Poisson bracket between
counter terms and the diffeomorphism constraint:
\begin{eqnarray} \label{SFPotTDiffeoPB}
\{H_{\varphi}^C[N],D[N^a]\} = \int\mathrm{d}^3x \bar{p}(\partial_c
\delta N^c \delta^a_i -\partial_i \delta N^a) \left[\delta
N\bar{p}^{3/2} f_4 V(\bar{\varphi}) \frac{\delta_a^i}{2\bar{p}}
\right. \nonumber\\ \left. + \bar{N}\bar{p}^{3/2}\left\{g_7
V_{,\varphi}(\bar{\varphi}) \delta\varphi
\frac{\delta_a^i}{2\bar{p}} + V(\bar{\varphi})\left( g_8
\frac{(\delta_c^j \delta E^c_j) \delta_a^i}{4\bar{p}^2} - g_9
\frac{\delta_c^i\delta_a^j\delta E^c_j} {2\bar{p}^2} \right)
\right\} \right] ~.
\end{eqnarray}
It is then easy to see that counter terms commute with the
diffeomorphism constraint only if the coefficients satisfy
\begin{equation}
\label{SFCTEqn2}
f_4 = 0  ~,~ g_7=0 ~,~ g_8 = 0 ~,~ g_9=0 ~.
\end{equation}
Thus diffeomorphism invariance of the counter terms allows just two
independent functions in the potential sector.  The non-vanishing
counter terms in the potential sector reduce to
\begin{equation} \label{SFPotCounterTerm}
H_{\varphi}^C[N] = \int\mathrm{d}^3x \left[ \delta N f_3
\bar{p}^{3/2} V_{,\varphi}(\bar{\varphi}) \delta\varphi + \bar{N}
\frac{1}{2} g_6 \bar{p}^{3/2}
V_{,\varphi\varphi}(\bar{\varphi}){\delta\varphi}^2 \right] ~.
\end{equation}
Contributions from counter terms of the potential sector to
the Poisson bracket between total Hamiltonians can be computed
as
\begin{eqnarray} \label{PotCTContribSFPot}
\mathcal{A}_{\varphi}^C &:=& \{H_{\varphi}^C[N_1],
H_{\rm grav}^P[N_2]+H_{\rm matter}^Q[N_2]\}
-(N_1\leftrightarrow N_2) \nonumber\\
&=&\int\mathrm{d}^3x \bar{N}(\delta N_1-\delta N_2) \left[
\frac{\bar{\nu}\bar{\pi}}{\bar{p}^{3/2}} \bar{p}^{3/2}
V_{,\varphi\varphi}(\bar{\varphi})\delta\varphi(f_3-g_6) +
\frac{\bar{\nu}\delta \pi}{\bar{p}^{3/2}} \bar{p}^{3/2}
V_{,\varphi}(\bar{\varphi}) (f_3) \right. \nonumber\\ && \left.
~~~~~~ + \frac{\bar{\nu}\bar{\pi}}{\bar{p}^{3/2}} \bar{p}^{3/2}
V_{,\varphi}(\bar{\varphi}) \frac{(\delta_c^j\delta
E^c_j)}{2\bar{p}}(-f_3) +
\frac{\bar{\alpha}\bar{k}}{\sqrt{\bar{p}}} \bar{p}^{3/2}
V_{,\varphi}(\bar{\varphi}) \delta\varphi
\left(2\bar{p}\frac{\partial f_3}{\partial\bar{p}}+3f_3\right)
\right] .
\end{eqnarray}
Combining (\ref{GravCTContribSFPot}), (\ref{KSCTContribSFPot}) and
(\ref{PotCTContribSFPot}) we form the total anomaly term
$\mathcal{A}_{\varphi} := \mathcal{A}_{{\rm
grav}\varphi}^{C}+\mathcal{A}_{\varphi\pi}^{C}+
\mathcal{A}_{\varphi}^C$ in the scalar matter potential sector:
\begin{eqnarray} \label{SFTotalPotAnomaly}
\mathcal{A}_{\varphi} &=& \int\mathrm{d}^3x \bar{N}(\delta
N_1-\delta N_2) \left[\bar{\nu}\bar{\pi}
V_{,\varphi\varphi}(\bar{\varphi})\delta\varphi \mathcal{D}_1 +
V_{,\varphi}(\bar{\varphi})\left\{ \bar{\nu}\delta \pi
\mathcal{D}_2 + \bar{\nu}\bar{\pi} \frac{(\delta_c^j\delta
E^c_j)}{2\bar{p}}\mathcal{D}_3\right\} \right. \nonumber\\ &&
\left. ~ + (\bar{\alpha}\bar{k}\bar{p})
V_{,\varphi}(\bar{\varphi})\delta\varphi \mathcal{D}_4 -
(\bar{\alpha}\bar{p})V(\bar{\varphi})\left\{ (\delta^c_j\delta
K_c^j)f +\bar{k} \frac{(\delta_c^j\delta E^c_j)}{2\bar{p}}(f+g)
\right\}\right]
\end{eqnarray}
where
\begin{eqnarray} \label{CTContribSFPotCoeff}
\mathcal{D}_1 &=& f_3-f_1-g_6 \nonumber\\
\mathcal{D}_2 &=& f_3-f_1+g_1 \nonumber\\
\mathcal{D}_3 &=& -f_3+f_2-f_1-g_2  \nonumber\\
\mathcal{D}_4 &=& 2\bar{p}\frac{\partial f_3}{\partial\bar{p}}+3f_3-3f
~.
\end{eqnarray}
As in the gravitational and matter kinetic sectors, coefficients of
the potential sector anomaly satisfy a non-trivial consistency
relation
\begin{equation}\label{D3D2Relation}
\mathcal{D}_3 = -\mathcal{D}_2
\end{equation}
using (\ref{SFCTEqn1}). We recall that counter terms
(\ref{PotCTContribSFPot}) of the potential sector have only two
unknown functions $f_3$ and $g_6$ which would be determined by
choosing, say, $\mathcal{D}_1=0$ and $\mathcal{D}_2=0$.  Then, we
would automatically have ${\cal D}_3=0$, but there are still
non-vanishing terms in the anomaly (\ref{SFTotalPotAnomaly}) of the
potential sector.
In particular, last two terms in (\ref{SFTotalPotAnomaly}) are
similar to the anomaly terms due to cosmological constant
(\ref{CTContribGravCC0}) and can be absorbed into anomaly terms of
the gravitational and matter kinetic sectors by subtracting the
total Hamiltonian constraint with suitable lapse function from it.
Vanishing of total anomaly terms from potential sector then
requires that ${\cal D}_4$ should also vanish {\em i.e.} ${\cal
D}_4=0$ \footnote{The anomaly term involving ${\cal D}_4$ can be
absorbed into anomaly terms of the gravitational and matter
kinetic sectors by suitably subtracting total Hamiltonian
constraint from it. However, such subtraction does not lead to a
consistent set of solutions for anomaly cancellation conditions.}.
However, requiring ${\cal D}_4$ to vanish imposes additional
restriction on counter terms which in turn requires primary
correction functions $\alpha$, $\nu$ and $\sigma$ to satisfy
another consistency requirement (see Eq.(\ref{alphanuRelation}))
in presence of a non-zero scalar potential, apart from the
relation (\ref{ANSRelation}).
After imposing ${\cal D}_1=0$, ${\cal D}_2=0$, ${\cal D}_3=0$ and
${\cal D}_4=0$, we can express remaining terms in equation
(\ref{SFTotalPotAnomaly}) as
\begin{eqnarray} \label{CTContribSFPotRemain0}
\mathcal{A}_{\varphi}^{\rm rem} =
\mathcal{A}_{{\rm grav}\varphi} + \mathcal{A}_{\pi\nabla\varphi}
- H^P\left[\bar{N}(\delta N_1-\delta N_2)
\frac{\bar{\alpha}}{\sqrt{\bar {p}}}
\left\{ (\delta^c_j\delta K_c^j)f +\bar{k}
\frac{(\delta_c^j\delta E^c_j)}{2\bar{p}}(f+g) \right\}
\right] ,
\end{eqnarray}
where part of the potential sector anomaly to be included in the
gravitational sector anomaly can be expressed as
\begin{eqnarray} \label{SFPotFeedbackAnoGrav}
\mathcal{A}_{{\rm grav}\varphi} = \frac{1}{8\pi G}
\int\mathrm{d}^3x \bar{N} (\delta N_1-\delta N_2) \bar{\alpha}^2
\left[ \bar{k}^2 (\delta^c_j\delta K_c^j) (- 3 f)
 +\bar{k}^3 \frac{(\delta_c^j\delta E^c_j)}{2\bar{p}} (-3f-3g)
\right] ~,
\end{eqnarray}
and another part that needs to be included in the anomaly expression
of the matter kinetic sector is
\begin{eqnarray}\label{SFPotFeedbackAnoKS}
\mathcal{A}_{\pi\nabla\varphi} = \int\mathrm{d}^3x \bar{N}
(\delta N_1-\delta N_2) \left[
\frac{\bar{\alpha}}{\sqrt{\bar{p}}}
\frac{\bar{\nu}\bar{\pi}^2}{2\bar{p}^{3/2}}
(\delta^c_j\delta K_c^j)(f)
+ \frac{\bar{\alpha}\bar{k}}{\sqrt{\bar{p}}}
\frac{\bar{\nu}\bar{\pi}^2}{2\bar{p}^{3/2}}
\frac{(\delta_c^j\delta E^c_j)}{2\bar{p}}
(f+g) \right] ~.
\end{eqnarray}

We have seen earlier that the presence of a non-zero cosmological
constant modifies the anomaly cancellation conditions as reflected in
equation (\ref{AnoCanGravCCEqn}). Similarly, the presence of a
non-trivial scalar matter potential leads to changes in the anomaly
cancellation conditions for both the gravitational sector as well as
the kinetic sector of matter. In particular, we can combine equations
(\ref{TotalGravAnomalyHH}) and (\ref{SFPotFeedbackAnoGrav}) to express
the total gravitational anomaly $\mathcal{A}_{\rm grav} :=
\mathcal{A}_{\rm grav}^{P} + \mathcal{A}_{\rm grav}^{C}+
\mathcal{A}_{{\rm grav}\varphi}$ as
\begin{equation} \label{TotalGravAnomalyHHWithSFPot}
\mathcal{A}_{\rm grav} = \frac{1}{8\pi G} \int\mathrm{d}^3x \bar{N}
(\delta N_1-\delta N_2)\bar{\alpha}^2 \left[ \frac{\bar{k}}{\bar{p}}
(\partial_c\partial^j\delta E^c_j) \mathcal{G}_1^{\varphi}
+ \bar{k}^2(\delta^c_j\delta K_c^j) \mathcal{G}_2^{\varphi}
+ \bar{k}^3 \frac{(\delta_c^j\delta E^c_j)}{2\bar{p}}
\mathcal{G}_3^{\varphi}\right],
\end{equation}
where the new coefficients are
\begin{eqnarray} \label{GPhiExpression}
\mathcal{G}_1^{\varphi} = \mathcal{G}_1 \quad,\quad
\mathcal{G}_2^{\varphi} = \mathcal{G}_2 -3f\quad,\quad
\mathcal{G}_3^{\varphi} = \mathcal{G}_3 - 3f - 3g ~.
\end{eqnarray}
Using equation (\ref{CTEqn1}), we again note that the new coefficients
satisfy the same non-trivial consistency relation
\begin{equation}\label{G3PhiG2PhiRelation}
\mathcal{G}_3^{\varphi} = - \mathcal{G}_2^{\varphi} ~.
\end{equation}
Thus also in the presence of a non-trivial scalar matter potential,
there are unambiguous solutions for $f$, $g$ and $h$ such that the
gravitational sector of
constraint algebra is anomaly-free. Similarly, for the matter kinetic
sector we can combine the original anomaly (\ref{SFHHAnomalySimp}),
contributions (\ref{GravCTContribSFKS}) from gravitational counter
terms, contributions (\ref{SFCTContribSFKS}) from kinetic sector
counter terms and contributions (\ref{SFPotFeedbackAnoKS}) from the
potential sector to express the total anomaly in the kinetic sector of
scalar matter as
\begin{eqnarray} \label{SFTotalKinAnomalyWithSFPot}
\mathcal{A}_{\pi\nabla} &:=&\mathcal{A}_{\rm matter}^P +
\mathcal{A}_{{\rm grav}\pi\nabla}^{C} + \mathcal{A}_{ \pi\nabla}^C +
\mathcal{A}_{\pi\nabla\varphi}\nonumber\\
&=& \int\mathrm{d}^3x\bar{N}
(\delta N_1-\delta N_2)
\left[\frac{\bar{\nu}\bar{\pi}\delta\pi}{\bar{p}^{3/2}}
\frac{\bar{\alpha}\bar{k}}{\sqrt{\bar{p}}}
\mathcal{B}_1^{\varphi} +
\frac{\bar{\nu}\bar{\pi}^2}{2\bar{p}^{3/2}}(\delta^c_j\delta K_c^j)
\frac{\bar{\alpha}}{\sqrt{\bar{p}}}
\mathcal{B}_2^{\varphi}
\right. \nonumber\\
&& \left. + \frac{\bar{\nu}\bar{\pi}^2}{2\bar{p}^{3/2}}
\frac{(\delta_c^j\delta E^c_j)}{2\bar{p}}
\frac{\bar{\alpha}\bar{k}}{\sqrt{\bar{p}}}
\mathcal{B}_3^{\varphi}
+
\frac{\bar{\nu}\bar{\pi}}{\bar{p}}\bar{\sigma}
\nabla^2\delta\varphi \mathcal{B}_4^{\varphi}
\right] ~,
\end{eqnarray}
with new coefficients
\begin{equation}\label{EqnBPhi}
\mathcal{B}_1^{\varphi} = \mathcal{B}_1 \quad,\quad
\mathcal{B}_2^{\varphi} = \mathcal{B}_2 +f \quad,\quad
\mathcal{B}_3^{\varphi} = \mathcal{B}_3 + f + g
\quad,\quad
\mathcal{B}_4^{\varphi} = \mathcal{B}_4
\end{equation}
in the matter kinetic sector anomaly expression. Also here,
the new coefficients satisfy a consistency relation
\begin{equation} \label{RelationBPhi123}
\mathcal{B}_3^{\varphi} = -2 \mathcal{B}_1^{\varphi} -
\mathcal{B}_2^{\varphi}
\end{equation}
using (\ref{CTEqn1}) and (\ref{SFCTEqn1}). This relation is analogous
to equation (\ref{RelationB123}) and likewise it ensures that anomaly
cancellation conditions lead to unambiguous expressions for counter
terms in the kinetic sector of scalar matter when there is a
non-trivial potential.

\subsection{Quantum corrected total Hamiltonian constraint}

We now combine all primary correction functions and the counter terms
to form the quantum corrected total Hamiltonian constraint $H^Q[N] :=
H^P[N] +H^C[N]$ for the system consisting of a scalar matter field
with arbitrary potential:
\begin{eqnarray} \label{HamGravQ}
H_{\rm grav}^Q[\bar{N}] &:=& \frac{1}{16\pi G} \int\mathrm{d}^3x
\bar{N}\left[ \bar{\alpha}{\mathcal H}^{Q(0)} +
\alpha^{(2)}{\mathcal H}^{Q(0)} +
\bar{\alpha}{\mathcal H}^{Q(2)}\right], \nonumber\\
H_{\rm grav}^Q[\delta N] &:=& \frac{1}{16\pi G}\int\mathrm{d}^3x
 \delta N \left[\bar{\alpha}{\mathcal H}^{Q(1)}\right] ~,
\end{eqnarray}
where the background density is unchanged
except for the explicit factor of the primary correction, i.e.\
${\mathcal H}^{Q(0)} \equiv {\mathcal H}^{(0)}$.
However the perturbed densities now involve counter terms:
\begin{eqnarray} \label{HamGravQDens}
\nonumber\\  {\mathcal H}^{Q(1)} &=& -4(1+f) \bar{k}\sqrt{\bar{p}}
\delta^c_j\delta K_c^j -(1+g)\frac{\bar{k}^2}{\sqrt{\bar{p}}}
\delta_c^j\delta E^c_j +\frac{2}{\sqrt{\bar{p}}}
\partial_c\partial^j\delta E^c_j  ~,
\nonumber\\  {\mathcal H}^{Q(2)} &=& \sqrt{\bar{p}} \delta
K_c^j\delta K_d^k\delta^c_k\delta^d_j - \sqrt{\bar{p}} (\delta
K_c^j\delta^c_j)^2 -\frac{2\bar{k}}{\sqrt{\bar{p}}} \delta
E^c_j\delta K_c^j
\\
&& \quad -\frac{\bar{k}^2}{2\bar{p}^{3/2}} \delta E^c_j\delta
E^d_k\delta_c^k\delta_d^j
+\frac{\bar{k}^2}{4\bar{p}^{3/2}}(\delta E^c_j\delta_c^j)^2
-(1+h)\frac{\delta^{jk} }{2\bar{p}^{3/2}}(\partial_c\delta E^c_j)
(\partial_d\delta E^d_k)  ~.\nonumber
\end{eqnarray}
It should be noted here that the terms in gravitational Hamiltonian which 
involve counterterms, contain either trace or divergence of perturbed 
basic variables. Thus, inclusion  of counterterms does not affect the 
earlier results for vector and tensor modes \cite{vector,tensor}.  
The complete quantum corrected matter Hamiltonian constraint is given by
\begin{eqnarray}\label{HamMatterQ}
 H^Q_{\rm matter}[\bar N]&=& \int_{\Sigma} \mathrm{d}^3x
\bar{N}\left[\left(\bar\nu\h_\pi^{Q(0)}+
\h_\varphi^{Q(0)}\right)+\left(\nu^{(2)}\h_\pi^{Q(0)}+
\bar\nu\h_\pi^{Q(2)}+\bar\sigma\h_\nabla^{Q(2)}+\h_\varphi^{Q(2)}\right)\right]~\nonumber\\
H_{\rm matter}^Q[\delta{N}] &=& \int\mathrm{d}^3x \delta N \left[
\bar{\nu}{\mathcal H}_\pi^{Q(1)} + {\mathcal H}_\varphi^{Q(1)}
\right] ~,
\end{eqnarray}
where background densities are again unchanged,
i.e.\ $\h_\pi^{Q(0)}\equiv\h_\pi^{(0)}$ and
$\h_\varphi^{Q(0)}\equiv\h_\varphi^{(0)}$.
As in the gravitational sector, perturbed matter densities
include counter terms and are given by
\begin{eqnarray}\label{HamMatterQDens}
{\mathcal H}_\pi^{Q(1)} &=& (1+f_1)\frac{\bar{\pi}
\delta{\pi}}{\bar{p}^{3/2}}
-(1+f_2)\frac{\bar{\pi}^2}{2\bar{p}^{3/2}} \frac{\delta_c^j
\delta E^c_j}{2\bar{p}}\\
 {\mathcal H}_\varphi^{Q(1)} &=&
\bar{p}^{3/2}\left( (1+f_3)V_{,\varphi}(\bar{\varphi})
\delta\varphi +V(\bar{\varphi}) \frac{\delta_c^j \delta
E^c_j}{2\bar{p}}\right)\nonumber\\
 \h^{Q(2)}_{\pi}&=&
(1+g_1)\frac{{\delta{\pi}}^2}{2\bar{p}^{3/2}}
-(1+g_2)\frac{\bar{\pi} \delta{\pi}}{\bar{p}^{3/2}}
\frac{\delta_c^j \delta E^c_j}{2\bar{p}}
\nonumber+\frac{1}{2}\frac{\bar{\pi}^2}{\bar{p}^{3/2}} \left(
(1+g_3)\frac{(\delta_c^j \delta E^c_j)^2}{8\bar{p}^2}
+\frac{\delta_c^k\delta_d^j\delta E^c_j\delta E^d_k}{4\bar{p}^2}
\right) \nonumber\\
\h^{Q(2)}_{\nabla}&=&\frac{1}{2}(1+g_5)\sqrt{\bar{p}}\delta^{ab}\partial_a\delta
\varphi
\partial_b\delta \varphi\nonumber\\
\h^{Q(2)}_{\varphi} &=&\bar{p}^{3/2} \left[(1+g_6)\frac{1}{2}
V_{,\varphi\varphi}(\bar{\varphi}) {\delta\varphi}^2 + V_{,\varphi}(\bar{\varphi})
\delta\varphi \frac{\delta_c^j \delta E^c_j}{2\bar{p}}+
V(\bar{\varphi})\left( \frac{(\delta_c^j \delta
E^c_j)^2}{8\bar{p}^2} -\frac{\delta_c^k\delta_d^j\delta
E^c_j\delta E^d_k}{4\bar{p}^2} \right)\right] \,. \nonumber
\end{eqnarray}

To summarize the conditions on non-vanishing coefficients of the
counter terms, we note that there are three such functions in the
gravitational sector (\ref{HamGravQDens}), six in the kinetic sector
and two in the potential sector of scalar matter
(\ref{HamMatterQDens}). Thus for the system under consideration we
have a total of eleven correction functions contained in all counter
terms. Invariance of counter terms under diffeomorphisms,
(\ref{CTEqn1}) and (\ref{SFCTEqn1}), led to four conditions
\begin{equation} \label{FinalConditionsDiffeo}
g = -2f  ~,~ f_2 = 2 f_1 ~,~ g_2=g_1 ~,~ g_3=2g_2
\end{equation}
among the non-vanishing coefficients.
These equations trivially lead to the solutions for $g$, $f_2$, $g_2$
and $g_3$, leaving seven functions to be determined.  

Cancellation of anomaly terms from the Poisson bracket between
Hamiltonian constraints led to three conditions (\ref{GPhiExpression})
from the gravitational sector
\begin{equation} \label{FinalConditionsGrav}
\mathcal{G}_1^{\varphi} = 0 ~,~
\mathcal{G}_2^{\varphi} = 0 ~,~
\mathcal{G}_3^{\varphi} = 0 ~,
\end{equation}
among which only two are independent due to
(\ref{G3PhiG2PhiRelation}). These two independent equations explicitly
solve $f$ and $h$ in terms of the primary correction function $\alpha$.
Thus, there are only five remaining functions that need to be
determined.  Anomaly cancellation from matter kinetic sector
(\ref{EqnBPhi}) leads to four conditions
\begin{equation}\label{FinalConditionsSFKE}
\mathcal{B}_1^{\varphi} = 0 ~,~
\mathcal{B}_2^{\varphi} = 0 ~,~
\mathcal{B}_3^{\varphi} = 0 ~,~
\mathcal{B}_4^{\varphi} = 0 ~.~
\end{equation}
Given the relation (\ref{RelationBPhi123}), there are only three
independent equations which leads to explicit solutions for $f_1$,
$g_1$ and $g_5$ in terms of primary correction functions $\alpha$ and
$\nu$. The remaining two free functions $f_3$ and $g_6$ are
constrained by requiring cancellation of anomalies from the potential
sector which gives four conditions
\begin{equation}\label{FinalConditionsSFPot}
\mathcal{D}_1 = 0 ~,~
\mathcal{D}_2 = 0 ~,~
\mathcal{D}_3 = 0 ~,~
\mathcal{D}_4 = 0 ~.
\end{equation}
Using equation (\ref{CTContribSFPotCoeff}), one notes that
$\mathcal{D}_1 = 0$ and $\mathcal{D}_2=0$ already determines both
$f_3$ and $g_6$ in terms of other counter terms coefficients which are
already fixed. 

While $\mathcal{D}_3 = 0$ is not an independent equation,
$\mathcal{D}_4 = 0$ imposes a non-trivial restriction on counter
terms. Since all of them have been determined at this stage,
consistency requires the primary correction functions to satisfy
\begin{equation}\label{alphanuRelation}
\f{\bar\alpha^\prime\bar p}{\bar\alpha} +\f{\bar
p}{3}\left(\f{\bar\alpha^\prime\bar p}{\bar
\alpha}\right)^\prime-\f{\bar\nu^\prime\bar p}{\bar \nu} -\f{\bar
p}{9}\left(\f{\bar\nu^\prime\bar p}{\bar
\nu}\right)^\prime+\f{2\bar p^2}{9}\left(\f{\bar\nu^\prime\bar
p}{\bar \nu}\right)^{\prime\prime}=0.
\end{equation}
Note that this relation ties the matter correction function to the
gravitational correction function, but it is independent of the matter
fields.
Finally, from (\ref{ANSRelation}) we have the relation
$\bar\alpha^2=\bar\nu\bar\sigma$ to be satisfied by the primary
correction functions.

To summarize, the requirement of anomaly
freedom in the constraint algebra tightly controls the allowed forms
of primary correction functions. 
For primary corrections of the form (\ref{AlphaHomo}), for instance,
one can easily see that solutions exist provided certain relations
between the coefficients $c_{\alpha}$, $c_{\nu}$ and powers
$n_{\alpha}$, $n_{\nu}$ in the two primary correction functions
$\alpha$ and $\nu$ are satisfied. Thus, quantization ambiguities are
non-trivially reduced, which allows stringent consistency tests by
direct calculations from a full representation of the underlying
operators.
These restrictions indirectly
help to eliminate some of the quantization ambiguities encountered
in quantizing inverse triad operators.

\section{Conclusions}

In this paper we have analyzed quantum corrected constraints at the
perturbative effective level. The key issue has been whether they form
a closed Poisson algebra, which would ensure consistency of equations
of motion they generate. We have found that a behavior of correction
functions $\alpha$, $\nu$ and $\sigma$ (scalars of zero density
weight) as in homogeneous models, which would suggest that they (i)
depend only on the triad $E_i^a$ (but not on the extrinsic curvature
$K_a^i$), (ii) depend on the triad algebraically (i.e. do not contain spatial
derivatives of $E_i^a$), and (iii) in the perturbed context, depend on the
background triad $\bar E_i^a$ and its perturbation $\delta E_i^a$ only
in the combination $\bar E_i^a+\delta E_i^a\equiv E_i^a$, would allow
a closed algebra only if the corrections are trivial.

At the same time, from the constructive point of view, it is not
surprising that the three conditions cannot be met together.
Indeed, the only scalar quantity that can be constructed from the
triad alone is its determinant --- a density weight one object ---
leaving no possibility of cancelling the density weight.  One could
relax any of these conditions and see whether that would allow
non-trivial corrections.  For instance, by allowing correction
functions to depend on spatial curvature, which would require spatial
derivatives of the triad, we could alleviate the problem of zero
density weight since this would make available the quantity
\[\frac{E_i^a}{\sqrt{|\det E|}}\D_a\left(\f{\D_b
E_i^b}{\sqrt{|\det E|}}\right)\,.
\]
On the other hand, if we relax the first condition, quantities of the
form $E_i^a K_a^i/\sqrt{|\det E|}$ would be allowed.

We were therefore led to conclude that expectations from homogeneous
models did not capture all possible quantum effects, and turned to
investigating what quantum corrections of inhomogeneous constraints
would be allowed in principle, and which ones should be ruled out. In
this process, we have generated several counterterms in addition to
the primary corrections suggested by homogeneous models. The resulting
counterterms admit non-trivial quantum corrections, and their presence
and form can be related to fundamental aspects of loop quantum
gravity. For instance, we have seen that quantum corrections must be
connection dependent even when they come from inverse triad
corrections.  This can be interpreted as meaning that the computation
of effective constraints, based on expectation values of constraint
operators, must be done in coherent states such that a holonomy
dependence of inverse triad expressions results.  Correction functions
must also depend on spatial derivatives of the triad, which can be
seen as leading terms in a derivative expansion of non-local
expressions involving fluxes, i.e.\ 2-dimensionally integrated triads.
There are also quite unexpected effects, such as counterterms in the
matter sector involving derivatives of the potential. Not all of them
are simply realized as a consequence of the expansion of
$V(\bar{\varphi}+\delta\varphi)$ by inhomogeneities. This suggests that the
matter potential must be quantized in a non-local way to ensure
anomaly freedom. This form of non-locality is currently not realized
in quantizations of scalar matter in loop quantum gravity.
It suggests concrete ways to change full constructions so as to
provide an (off-shell) anomaly-free formulation.

{}From the perturbed second order constraints one can directly derive
Hamiltonian equations of motion for the perturbed variables as well as
gauge transformations on them. Both equations of motion and gauge
transformations will be corrected by quantum gravity terms, which has
to be combined for equations of motion of gauge invariant variables of
the form (\ref{PertI})--(\ref{PertIII}).  Imposing the conditions
found for an anomaly-free constraint algebra must, on general grounds,
result in a consistent set of equations. This has already been
verified for vector and tensor modes (see \cite{vector} and
\cite{tensor} respectively). In a companion paper
\cite{ScalarGaugeInv} we explicitly derive gauge transformations and
construct gauge invariant variables taking into account quantum
corrections, which we will then use to derive gauge invariant
equations of motion describing cosmological perturbations.

We have provided one consistent set of equations by a process which
demonstrates that possibilities of non-trivial quantum corrections are
rather tight.
In fact, existing proposals for primary correction functions are
non-trivially restricted.
Yet, different versions may be available, which could in principle be
compared with full derivations of effective Hamiltonians to fix
remaining ambiguities. But there may also be quantization ambiguities
which cannot be removed based solely on consistency considerations;
they would have to be restricted phenomenologically instead. It is
thus important also for a fundamental understanding to evaluate
cosmological implications of the quantum corrected perturbation
equations.

In addition to other choices regarding one type of quantum
corrections, which in this paper is inverse triad corrections, there
are different general types of corrections. 
In loop quantum gravity, we have two additional classes: corrections
of higher powers of the connection or extrinsic curvature due to the
use of holonomies, and genuine quantum back-reaction effects which
include the influence of the whole wave function on its expectation
values. (It is the latter which underlies constructions such as the
low-energy effective action used in particle physics.)  These
corrections turn out to be more difficult to compute in consistent
form, which is still in progress. Our consistent equations are thus
not to be considered as complete effective equations, and including
the corresponding terms of one type may add to the effects of another
type or decrease them (in a way which is regime dependent).
But it is unlikely that complete cancellations happen because
corrections of the different types take so different forms. Moreover,
a complete cancellation would mean that the characteristic fundamental
representation of loop quantum gravity would leave no trace on the
physics of the theory. 

While quantitative results are expected to depend on the specific form
of corrections and the interplay of different types, the occurrence of
qualitative effects signalling deviations from classical relativity is
more robust. This differs from other results of loop quantum
cosmology, such as bounces in homogeneous models where a sharp
zero-result for the time derivative of a scale factor is required.
Such sharp conditions can easily be destroyed when additional quantum
corrections are included; see e.g.\ \cite{BounceSqueezed}. Compared to
that, the complete elimination of qualitative effects of one type of
correction by including another type is highly unlikely.

The consistent constraint algebra shows that non-trivial quantum
corrections which reflect the underlying discreteness of spatial
geometry are possible. In this sense, general covariance is preserved.
However, we have shown that the classical constraint algebra, while
consistently deformed, is not represented exactly but receives quantum
corrections from the corrected constraints. One can see this directly
from (\ref{HHQuant}), for instance, which carries a factor of
$\bar{\alpha}^2$ in the smearing function of the diffeomorphism
constraint.  This is required by consistency since the classical
algebra cannot be realized with non-trivial quantum corrections. Thus,
an effective action of loop quantum gravity cannot be simply of
higher-curvature type.  (Non-local features, for instance, would then
be essential.)  Nevertheless, we expect that some of the corrections
can be formulated by effective higher-curvature actions which applies
even to the inverse triad corrections used here. Some of the
counterterms, which depend on extrinsic curvature components as well
as spatial derivatives of the triad, can in fact be interpreted as
bringing the corrected constraints in a form amenable to being formulated
as a higher-curvature action. In this context, we emphasize that the
absence of new degrees of freedom in this Hamiltonian framework is
not in conflict with the higher-derivative nature of
higher-curvature effective actions as also discussed in
\cite{Karpacz}: a perturbative interpretation of higher-derivative
actions, which is the only appropriate way in quantum gravity, does
not give rise to more solutions than expected classically
\cite{Simon}.

\section*{Acknowledgements}

This work was supported in part by NSF grants PHY0653127 and PHY0456913.
MK was supported by NSF grant PHY0114375.
SS wishes to thank Roy Maartens and Kevin Vandersloot for discussions.
He is being supported by the Marie Curie Incoming International Grant
IIF-2006-039205.

\section*{Appendix}

\begin{appendix}

\section{Comparison with isotropic models}
\label{a:iso}

An important measure for the size of quantum corrections is the
characteristic scale $a_*$ which signals the onset of non-perturbative
effects. As a critical value for the scale factor, it does not have
absolute meaning because it can be rescaled by a choice of
coordinates. It is ratios such as $a/a_*$ which have physical meaning
related to the patch density of a quantum gravity state.  For a denser
state features of correction functions based on inverse triad
components are realized on larger scales, which increases the
corresponding quantum corrections.

These effects are also important if one tries to include the behavior
in homogeneous models, even though an exactly homogeneous model
provides only limited means of referring to spatial discreteness of an
underlying state and its refinement. For this reason, care is needed
if one tries to address possibilities of refinement schemes and the
size of quantum corrections in purely homogeneous settings, as is
often done due to the simplicity of homogeneous models.  Quantum
corrections in a fully inhomogeneous situation must be expected to be
larger than in a naive isotropic quantization which ignores the factor
${\cal N}$ of the patch density and implicitly assumes ${\cal
N}^{1/3}\sim 1$ as in \cite{Bohr}.  This is the reason why some
minisuperspace considerations artificially suppress those corrections.
Corrections from holonomies, on the other hand, increase with
decreasing vertex density such that they would appear to be more
pronounced. It is possible to mimic the enhancement of inverse triad
corrections even in exact homogeneous models by computing their
correction function for operators based on higher representations of
SU(2) instead of the fundamental one \cite{Ambig}. The corresponding
spin label is then related to the vertex density.

In addition to the size of corrections, there is also the issue of the
correct scaling behavior of correction terms. To have independence of
the coordinate size $V_0$ of the region whose patches are counted, we
must have ${\cal N}\propto V_0$. However, this provides coordinate
independent quantum corrections only if we multiply by another
function which can absorb the coordinate dependence of $V_0$. The
simplest possibility in isotropic models is to use ${\cal N}\propto
a^3V_0$ for which corrections depend neither on the size of the volume
nor on coordinates. This behavior is indeed well-motivated based on
lattice refinements (where the physical vertex density is constant)
and was introduced in \cite{APSII} based on scaling arguments. If
${\cal N}$ is allowed to depend directly only on the scale factor (and
not on $\dot{a}/a$, say), and there is no other parameter which
rescales under changing coordinates other than $V_0$, this is indeed
the only consistent choice. In this sense, the behavior proposed in
\cite{APSII} is unique in isotropic models up to a single constant
which determines the absolute number of patches.

This uniqueness is, however, contingent on conditions which are too
strong for reliably modelling what happens in inhomogeneous
situations. If the model is no longer exactly homogeneous, the
refinement of underlying states is history dependent in ways which do
not simply amount to a dependence on $a$. One can always express the
refinement as a dependence on scaling-independent observables such as
$\dot{a}/a$ which provide an equally good measure for the history of
different phases of the universe. For a given background solution, one
can then express this as a dependence on $a$ alone, given that all
observables depend on $a$. In general, however, this provides more
complicated functions than just ${\cal N}(a)\propto a^3V_0$ for the
patch density. Some phases can, for instance, be described by a
power-law form ${\cal N}(a)={\cal N}_0a^xV_0$ where ${\cal N}_0$
arises in a complicated process by expressing the refinement via a
function only of $a$. In particular, because the original refinement
is history dependent only via observable quantities, this constant
will automatically be equipped with a scaling dependence such that
${\cal N}_0a^xV_0$ is coordinate independent even if $x\not=3$. The
emergence of such a parameter can only be seen in the proper
inhomogeneous context, invalidating considerations based solely on
homogeneous models.

\section{Poisson brackets between unperturbed constraints.}
\label{a:Unpert}

It is instructive to compute Poisson brackets between primary quantum
corrected constraints without expanding by inhomogeneities.  We first
consider the gravitational Hamiltonian constraint
(\ref{HamConstQuant}). As $H_{\rm grav}^P$ commutes with matter 
diffeomorphism $D_{\rm matter}$, it is sufficient to compute $\{H^P_{\rm grav}[N], 
D_{\rm grav}[N^a]\}$. Classically, this Poisson bracket is
\begin{equation} \label{FullHDClassical}
\{H_{\rm grav}[N], D_{\rm grav}[N^a]\} = - H_{\rm grav}[N^a\partial_a N] ~.
\end{equation}

The (gravitational) diffeomorphism constraint acts as a Lie
derivative on a (gravitational) phase space function
\begin{equation} \label{Grav_Diff_Lie}
\left\{F(A,E),D_{\rm grav}[N^a]\right\}={\mathcal L_{N^a}F}~.
\end{equation}
The Poisson bracket between $H_{\rm grav}^P[N]$ and $D_{\rm
grav}[N^a]$ then is
\begin{eqnarray}
\{H_{\rm grav}^P[N], D_{\rm grav}[N^a]\} &\equiv&
\{\int_{\Sigma}\mathrm{d}^3x N \alpha{\mathcal H}, D_{\rm grav}[N^a]\}
=  \int_{\Sigma}\mathrm{d}^3x N {\mathcal L}_{N^a}(\alpha{\mathcal H})
\nonumber\\
&=& \int_{\Sigma}\mathrm{d}^3x N \left(N^a \partial_a(\alpha{\mathcal H})+
(\partial_a N^a)\alpha{\mathcal H}\right)
= \int_{\Sigma}\mathrm{d}^3x (- N^a\partial_a N) \alpha{\mathcal
H} \nonumber\\
&\equiv& - H_{\rm grav}^P[N^a\partial_a N] ~.\label{HD}
\end{eqnarray}
Where we have used the fact that the quantity $\alpha {\mathcal H}$
has density weight one to expand the Lie derivative and integrated by
parts in the next line. By the same token, we would not obtain the
correct algebra if $\alpha$ would not be of density weight
zero. Clearly, any functional of the (gravitational) variables defined
by integration must be an integral of a density-weight-one
function. Thus the only restriction on the correction function,
obtained so far, is that it must be of zero density weight. In fact,
only this condition makes the spatial integrals well-defined.

The matter Hamiltonian constraint has non-zero Poisson brackets with
both gravitational and matter parts of the diffeomorphism constraint.
However, the total diffeomorphism constraint acts as a Lie derivative
on any function of {\em all} phase space variables
\begin{equation} \label{Diff_Lie}
\left\{F(A,E,\varphi,\pi),D[N^a]\right\}={\mathcal L_{N^a}F}~.
\end{equation}
Hence its Poisson bracket with the (gravitational) diffeomorphism
constraint should boil down to an expression analogous to
(\ref{HD}). Again, the only condition on the correction functions is
that they have zero density weight. In that case not only are the
quantum corrected constraints first class, but also form an algebra
identical to the classical one.

In what follows we will make extensive use of

\begin{lemma}\label{lemma1}
Consider a functional
\begin{equation}\label{F_N}%
F[N]=\int{\!\md^3 x\, N(x) f\!\left(\varphi,\pi\right)}
\end{equation}%
of two canonically conjugate scalar\footnote{Generalization to
the case of tensorial as well as several conjugate fields is
straightforward.} fields $\varphi$ and $\pi$. If $f$ does not
depend on spatial derivatives of the fields, the Poisson bracket
\begin{equation}\label{F_N1N2}%
\{F[N_1],F[N_2]\}_{(\varphi,\pi)} = 0\nonumber
\end{equation}%
vanishes
\end{lemma}
\begin{proof}
Since the integrand does not contain spatial
derivatives, we have the functional derivative
$\delta F[N]/\delta \varphi = N\D f/\D \varphi$
and
\begin{eqnarray}%
\{F[N_1],F[N_2]\} &\equiv& \int{\!\md^3 x\, \left(\frac{\delta
F[N_1]}{\delta \varphi}\frac{\delta F[N_2]}{\delta \pi} -
(N_1\leftrightarrow N_2)\right)}\nonumber\\
&=&\int{\!\md^3 x\, \left(N_1 \frac{\D f}{\D \varphi}N_2\frac{\D
f}{\D \pi} - (N_1\leftrightarrow N_2)\right)}
=0\nonumber
\end{eqnarray}%
\end{proof}

On the contrary, if spatial derivatives of the fields are present in
the integrand, the relevant functional derivative involves derivatives
of the smearing function which implies a non-vanishing final
expression for the Poisson bracket after anti-symmetrization over
$N_1$ and $N_2$.

Using this result, let us analyze the expression
\begin{eqnarray}
\{H^P[N_1],H^P[N_2]\} &\equiv&
\{H^P_{\rm grav}[N_1]+H^P_{\rm matter}[N_1],H^P_{\rm grav}[N_2]+
H^P_{\rm matter}[N_2]\} \nonumber \\
&=&\{H^P_{\rm grav}[N_1],H^P_{\rm grav}[N_2]\}+\{H^P_{\rm matter}[N_1],
H^P_{\rm matter}[N_2]\}\\
&+&\left(\{H^P_{\rm grav}[N_1],H^P_{\rm matter}[N_2]\}-(N_1 \leftrightarrow
N_2)\right) \nonumber
\end{eqnarray}
term by term. The gravitational constraints yield
\begin{eqnarray}\label{HGwithHG}
\{H^P_{\rm grav}[N_1],H^P_{\rm grav}[N_2]\}&=&\int_{\Sigma}\mathrm{d}^3x
\left(\frac{\delta H_{\rm grav}[\tilde N_1]}{\delta A_a^i} \frac{\delta
H_{\rm grav}[\tilde N_2]}{\D E_i^a} - (N_1 \leftrightarrow
N_2)\right)\nonumber\\
&+&\int_{\Sigma}\mathrm{d}^3x \left(\frac{\delta H_g[\tilde
N_1]}{\delta A_a^i} \frac{\D \alpha}{\D E_i^a} N_2 \h - (N_1
\leftrightarrow
N_2)\right) \nonumber \\
&=&D_{\rm grav}[\tilde N_1 \D^a \tilde N_2 - \tilde N_2 \D^a \tilde
N_1]+{\mathcal A}_{\rm grav\,grav},
\end{eqnarray}
where we have used the fact that if $\tilde{N}_{1,2}$ were independent
of phase space variables then we would simply have the classical
constraint algebra but with new lapse functions $\tilde{N}_1$ and
$\tilde{N}_2$. However, since $\tilde{N}_{1,2}$ do depend on the
densitized triad there is an extra (potentially anomalous) term in the
Poisson bracket which is the second term ${\mathcal A}_{\rm grav\,grav}$,
proportional to the derivatives of the correction function. The
non-trivial contributions to the anomaly
\begin{eqnarray}\label{HwithH_anomaly}
{\mathcal A}_{\rm grav\,grav}&=&-\int_{\Sigma}\mathrm{d}^3x
\left\{N_1\frac{\D \alpha}{\D E_{i^\prime}^a}\h\frac{\delta}{\delta
A_a^{i^\prime}}\left(\int_{\Sigma}\mathrm{d}^3y N_2 \alpha \frac{2
E_i^c E_j^d}{\sqrt{\det E}} \D_c A_d^k \epsilon_{ijk}\right) - (N_1
\leftrightarrow N_2)\right\}\nonumber \\ &=&\int_{\Sigma}\mathrm{d}^3x
\, \h\frac{\D \alpha}{\D E_{k}^a}\epsilon_{ijk}\left\{N_1 \D_c
\left(\alpha N_2\frac{2 E_i^c E_j^a}{\sqrt{\det E}}\right) - (N_1
\leftrightarrow N_2)\right\}
\end{eqnarray}
come from the gradient terms of the Hamiltonians. Note that, for
convenience, we switched the order of terms in the first line of
(\ref{HwithH_anomaly}). In the second line, the only term in the
parenthesis that survives the anti-symmetrization is the one
proportional to the gradient of the lapse function $\D_c N_2$. Thus
the anomaly simply boils down to
\begin{equation}\label{HwithH_anomaly1}
{\mathcal A}_{\rm grav\,grav}=H_{\rm grav}^P[M_\alpha] \quad {\rm with}\quad
M_\alpha:=2\epsilon_{ijk}\frac{E_i^c E_j^a}{\sqrt{\det E}}\frac{\D
\alpha}{\D E_k^a}(N_1\D_c N_2 - N_2\D_c N_1).
\end{equation}
It is easy to see that the symmetricity condition
\begin{equation}\label{cond_on_alpha}
E_j^a\frac{\D \alpha}{\D E_k^a}=E_k^a\frac{\D \alpha}{\D E_j^a}
\end{equation}
is sufficient to make the anomaly (\ref{HwithH_anomaly1}) vanish
due the contraction with $\epsilon_{ijk}$. We should point out
that (\ref{cond_on_alpha}) is definitely satisfied for any
triad-dependent scalar function, which has all internal indices
contracted.

The cross Poisson bracket $\{H_{\rm grav}^P[N_1],H_{\rm
matter}^P[N_2]\}-(N_1\leftrightarrow N_2)$ can be computed
similarly. In the absence of curvature couplings, such that the matter
Hamiltonian contains neither connection nor spatial derivatives of
the triad, this Poisson bracket is given by
\begin{eqnarray}\label{HGwithHM}
\{H^P_{\rm grav}[N_1],H^P_{\rm matter}[N_2]\}-(N_1 \leftrightarrow N_2)&=&
\int_{\Sigma}\mathrm{d}^3x \left(\frac{\delta H_{\rm grav}[\tilde
N_1]}{\delta A_a^i} \frac{\D\nu}{\D E_i^a} N_2 \h_\pi - (N_1
\leftrightarrow N_2)\right) \nonumber \\
&+& \int_{\Sigma}\mathrm{d}^3x \left(\frac{\delta H_{\rm grav}[\tilde
N_1]}{\delta A_a^i} \frac{\D\sigma}{\D E_i^a} N_2 \h_\nabla - (N_1
\leftrightarrow N_2)\right) \nonumber \\
&=&H_\pi^P[M_\nu]+H_\nabla^P[M_\sigma],
\end{eqnarray}
where
\begin{equation}\label{HGwithHM_anomaly1}
H_\pi^P[M_\nu]= \int_{\Sigma}\mathrm{d}^3x M_\nu \h_\pi, \quad
H_\nabla^P[M_\sigma]= \int_{\Sigma}\mathrm{d}^3x M_\sigma
\h_\nabla
\end{equation}
with the effective lapse functions
\begin{equation}
M_\nu:=2\epsilon_{ijk}\frac{E_i^c E_j^a}{\sqrt{\det E}}\frac{\D
\nu}{\D E_k^a}(N_1\D_c N_2 - N_2\D_c N_1), \quad
M_\sigma:=2\epsilon_{ijk}\frac{E_i^c E_j^a}{\sqrt{\det E}}\frac{\D
\sigma}{\D E_k^a}(N_1\D_c N_2 - N_2\D_c N_1)
\end{equation}
similar to (\ref{HwithH_anomaly1}). These vanish if the correction
functions satisfy
\begin{equation}\label{cond_on_nu-sigma}
E_j^a\frac{\D \nu}{\D E_k^a}=E_k^a\frac{\D \nu}{\D E_j^a}, \quad
E_j^a\frac{\D \sigma}{\D E_k^a}=E_k^a\frac{\D \sigma}{\D E_j^a}.
\end{equation}

Finally, the Poisson bracket between two matter Hamiltonians
involves only functional derivatives with respect to the matter
variables $\varphi$ and $\pi$. By virtue of lemma (\ref{lemma1}),
the non-trivial contribution comes from
\[\left\{\int_{\Sigma}\mathrm{d}^3x N_1 \nu
\h_\pi,\int_{\Sigma}\mathrm{d}^3x N_2 \sigma
\h_\nabla\right\}_{(\varphi,\pi)}-(N_1 \leftrightarrow N_2).
\]
Since the correction functions do not depend on the matter
variables, they act as constant factors, i.e.
\begin{eqnarray}\label{HMwithHM}
\{H^P_{\rm matter}[N_1],H^P_{\rm matter}[N_2]\}=
D_{\rm matter}[\nu\sigma (N_1 \D^a N_2 - N_2 \D^a
N_1)],
\end{eqnarray}
Combining (\ref{HGwithHG}), (\ref{HGwithHM}) and (\ref{HMwithHM})
and assuming (\ref{cond_on_alpha}) and
(\ref{cond_on_nu-sigma}) we obtain
\begin{eqnarray}\label{HwithH}
\{H^P[N_1],H^P[N_2]\}=D_{\rm grav}[\alpha^2 (N_1 \D^a N_2 - N_2 \D^a
N_1)]+D_{\rm matter}[\nu\sigma (N_1 \D^a N_2 - N_2 \D^a N_1)].
\end{eqnarray}
It is easy to see that the constraint algebra closes, if
$\alpha^2=\nu\sigma$ in addition to the requirement that $\alpha$,
$\nu$ and $\sigma$ are all scalars of vanishing density weight. In
that case the right hand side of (\ref{HwithH}) reduces to the total
diffeomorphism constraint
\begin{eqnarray}\label{HwithH_nonanom}
\{H^P[N_1],H^P[N_2]\}=D[\alpha^2 (N_1 \D^a N_2 - N_2 \D^a
N_1)]\equiv D[\tilde N_1 \D^a \tilde N_2 - \tilde N_2 \D^a \tilde
N_1)].
\end{eqnarray}

So far in this appendix, we have worked non-perturbatively which gives
only a few conditions on quantum correction functions.  The anomaly
freedom conditions (\ref{HDAnomalyFreeCond}) and
(\ref{HHAnomalyFreeCond}) obtained in the main part of this paper,
where the condition of vanishing density weight turns out to be quite
non-trivial, appear much more restrictive compared with the relatively
mild-looking requirement derived in the context of the unperturbed
system in this appendix. It is therefore pertinent to comment on this
apparent discrepancy.

Note that the conditions on the three correction functions imply the
same functional form of $\alpha$, $\nu$ and $\sigma$. Thus we shall
restrict our consideration to only one of them. In section
\ref{PrimCorr}, we had made the following assumptions concerning the
primary correction function $\alpha$:\\ (i) $\alpha$ depends only on
the triad $E_i^a$ (but not on the extrinsic curvature $K_a^i$ or the
connection),\\ (ii) $\alpha$ depends only algebraically on the triad
$E_i^a$ (but not its spatial derivatives)\\ (iii) in the perturbed
context, $\alpha$ depends on the background triad $\bar E_i^a$ and its
perturbation $\delta E_i^a$ only in the combination $\bar E_i^a+\delta
E_i^a\equiv E_i^a$ (i.e. $\alpha$ is expected to originate from a full
unperturbed expression).

One can check by inspection that assumption (iii) implies that
(\ref{HHAnomalyFreeCond}) is automatically satisfied. Indeed,
using the Taylor expansion
\be%
 \alpha(E_i^a) = \alpha(\bar E_i^a)+\left.\f{\D \alpha(E_i^a)}{\D
 E_i^a}\right|_{\bar E_i^a} \delta E_i^a +
\left.\f{1}{2} \f{\D^2 \alpha(E_i^a)}{\D E_i^a \D
 E_j^b}\right|_{\bar E_i^a}\delta E_i^a \delta E_i^b+\cdots \equiv
 \bar\alpha+\alpha^{(1)}+\alpha^{(2)}+\cdots
 \ee it is easy to see that the terms on the right hand side are not
 entirely independent. Clearly the relations between $\bar\alpha$,
 $\alpha^{(1)}$, and $\alpha^{(2)}$ are exactly written in
 Eq.~(\ref{HHAnomalyFreeCond}).

However, of greater concern is the other condition,
Eq.~(\ref{HDAnomalyFreeCond}). In particular, the requirement
$\alpha^{(1)}=0$ (along with (\ref{HHAnomalyFreeCond})) rules
out all possible non-trivial solutions. In order to understand its
origin let us revisit the seemingly trivial restriction on the
correction function to be of zero density weight. We start by
formulating the following

\begin{lemma}
A scalar $\alpha(E_i^a)$ of density weight zero satisfying
the three assumptions above must be a constant function.
\end{lemma}

\begin{proof}
Consider a scalar $\alpha(E_i^a)$ of density weight
$w$ satisfying the aforementioned assumptions. Its Lie derivative
along an arbitrary shift vector $N^a$ is given by
\[
\l_{\vec N}\alpha=N^b \D_b\alpha+w\alpha \D_b N^b.
\]
On the other hand,
\[
\l_{\vec N}\alpha=\f{\D \alpha}{\D E_i^a}\l_{\vec N} E_i^a=\f{\D
\alpha}{\D E_i^a} \left( N^b \D_b E_i^a - E_i^b \D_b N^a +E_i^a
\D_b N^b\right).
\]
These equations are valid for any $N^a$. In the context of
cosmological perturbation theory, $\bar N^a=0$, hence $N^a=\delta
N^a$. In the perturbative expansion of the right hand side of the
equations there is no contribution from the background part.
Equating the corresponding linear order terms, we obtain
\[
w \bar \alpha \D_b \delta N^b=\left(\f{\D \alpha}{\D
E_i^a}\right)^{(0)}\left(\bar E_i^a \D_b \delta N^b-\bar E_i^b
\D_b \delta N^a \right).
\]
Using $\bar E_i^a=\bar p \delta_i^a$ the derivative $(\D
\alpha/\D E_i^a)^{(0)}\equiv (\D \alpha/\D \bar
E_i^a)|_{\bar E_i^a}$ can be rewritten as $\f{1}{3}\delta_i^a\D
\bar \alpha/\bar p$, which yields
\[
\f{2\bar p}{3}\f{\D \bar \alpha}{\D \bar p} \D_b \delta N^b =w
\bar \alpha \D_b \delta N^b.
\]
The divergence of a generic shift vector does not vanish, and
therefore the derivative of the background correction function is $\D
\bar \alpha/\D \bar p=\f{2}{3}\bar p w \alpha$.  Requiring $w=0$
results in $\D \bar \alpha/\D \bar p=0$ and consequently from
(\ref{HHAnomalyFreeCond}), $\alpha^{(1)}=0$, $\alpha^{(2)}=0$ and so
on. This concludes the proof of the lemma.
\end{proof}

In the light of this we are led to the following conclusion. The three
assumptions that we made on the functional form of the correction
functions are incompatible with the conditions for anomaly freedom,
unless $\alpha$, $\nu$ and $\sigma$ are constants.  Therefore, to allow
non-trivial solution we have to relax one or more of the assumptions
which makes the algebra much more involved. In the main text of this
paper, we organize these calculations by the method of counterterms.

\section{Poisson brackets of perturbed variables}
\label{a:Poisson}

A direct application of the Poisson brackets given by (\ref{PB_full})
can sometimes be problematic. For instance, the Poisson bracket
between the two original fields $\{\varphi,\pi\}$, given by
\bq \label{PB_Phi_Pi}
\{\bar\varphi+\delta\varphi(x),\bar\pi+\delta\pi(x)\}_{\bar\varphi,
\bar\pi,\delta\varphi,\delta\pi}&=&
\{\bar\varphi,\bar\pi\}_{\bar\varphi,\bar\pi}+
\{\delta\varphi(x),\delta\pi(y)\}_{\delta\varphi,\delta\pi},\nonumber\\
&=&\f{1}{V_0}+\delta(x-y),
\eq
does not agree with the original expression
$\{\varphi(x),\pi(y)\}=\delta(x-y)$. This can be traced to the fact that
(\ref{PB_full}) provides Poisson brackets for the fields
$(\bar\varphi,\bar\pi,\delta\varphi,\delta\pi)$ only if the conditions
(\ref{2nd_class_constraints}) are used in (\ref{New_Sim_Str}) to
identify $\bar\varphi$ and $\bar\pi$ with the sole zero modes of inhomogeneous
fields.  The constraints (\ref{2nd_class_constraints}) clearly have a
nonzero Poisson bracket $\{\chi_1,\chi_2\}$, which makes them of the
second class.

According to Dirac \cite{DirQuant}, second class constraints
correspond to non-physical degrees of freedom and can be dealt
with in the following way. i) One should take linear combinations
of (all) the constraints, in order to bring as many of them into
first class form as possible and ii) redefine the Poisson bracket to
\be \label{DB}
\{F,G\}^*_{\delta\varphi,\delta\pi}=
\{F,G\}_{\delta\varphi,\delta\pi}-\{F,\chi_a\}_{\delta\varphi,\delta\pi}C^{ab}\{\chi_b,G\}_{\delta\varphi,\delta\pi},
\ee
where
\[
C^{ab}\{\chi_b,\chi_c\}=\delta^a_c \nonumber
\]
so as to remove the variations with respect to the non-physical
degrees of freedom. Using (\ref{2nd_class_constraints}) we obtain
\[
C^{11}=C^{22}=0, \quad C^{21}=-C^{12}=(V_0\lambda_1\lambda_2)^{-1},
\]
which implies
\be\label{DB_Def}
\{F,G\}^*_{\delta\varphi,\delta\pi}=
\{F,G\}_{\delta\varphi,\delta\pi}-\f{1}{V_0}\left(\int{\md^3 z\f{\delta
F}{\delta (\delta \varphi)}}\int{\md^3 z^\prime\f{\delta G}{\delta
(\delta \pi)}} - (F\leftrightarrow G)\right)
\ee

Let us first point out the basic properties of the Dirac bracket
(\ref{DB_Def}). For the field perturbations
\[
\{\delta\varphi(x),\delta\pi(y)\}^*_{\bar\varphi,\bar\pi,
\delta\varphi,\delta\pi}=\delta(x-y)-\f{1}{V_0}\,.
\]
Clearly the last term would remove the extra contribution in
(\ref{PB_Phi_Pi}) yielding the expected result
\bq \label{DB_Phi_Pi}
\{\varphi,\pi\}^*_{\bar\varphi,\bar\pi,\delta\varphi,\delta\pi}&=&
\{\bar\varphi,\bar\pi\}_{\bar\varphi,\bar\pi}+
\{\delta\varphi(x),\delta\pi(y)\}^*_{\delta\varphi,\delta\pi}
=\f{1}{V_0}+\delta(x-y)-\f{1}{V_0}
=\{\varphi,\pi\}_{\varphi,\pi}\,.\nonumber
\eq
Thus the Dirac bracket ensures a correct transition from the full
theory to the perturbed one. By construction, the constraints
(\ref{2nd_class_constraints}) now commute, $\{\chi_1, \chi_2\}^*=0$,
and we can impose $\chi_1$ and $\chi_2$ strongly. Moreover,
Dirac brackets between a first order functional and a functional of
arbitrary order vanish, which can be seen by inspection using
(\ref{DB_Def}). Thus for any two functionals, their linear terms do
not contribute to the Dirac bracket.\footnote{We are mostly
interested in linearized equations of motion, i.e. in the functionals
(constraints) we should keep terms up to the second order.}

We will now address an issue directly related to closure of the
constraints algebra. When computing Poisson brackets in the
context of perturbation theory, one has a choice between two
methods:

1) Calculate the Poisson bracket of the constraints with respect
to the full fields and expand the resulting expression in orders
of perturbations.

or

2) Expand the constraints first and then compute their Poisson
(Dirac) brackets in terms of the expanded fields.

It is, in general, not guaranteed that the two approaches agree for
arbitrary functionals which depend on the fields and their first
derivatives. However, we have

\begin{lemma}
If the fields $\bar\varphi$, $\bar\pi$, $\delta\varphi$ and
$\delta\pi$ enter the functionals
\be\label{F&G}%
F=\int{\!\md^3 x\,
f\!\left(\varphi,\nabla\varphi,\pi,\nabla\pi\right)},\quad
G=\int{\!\md^3 x\, g\!\left(\varphi,\nabla\varphi,\pi,\nabla\pi\right)}
\ee%
only as a combination $\varphi\equiv \bar\varphi+\delta\varphi$ or
$\pi\equiv \bar\pi+\delta\pi$, then the two procedures described
above yield the same result for $\{F,G\}$.
\end{lemma}

\begin{proof} We shall show that the second procedure is equivalent
to the first one. First of all, as linear terms do not contribute
to (\ref{DB_Def}), we can rewrite the Dirac bracket between two
expanded constraints as
\bq
\{F^{(0)}+F^{(2)},G^{(0)}+G^{(2)}\}^*_{\bar\varphi,\bar\pi,
\delta\varphi,\delta\pi}&=&
\{F^{(0)}+F^{(1)}+F^{(2)},G^{(0)}+G^{(1)}+G^{(2)}\}^*_{\bar\varphi,\bar\pi,
\delta\varphi,\delta\pi}\nonumber
\\
&\equiv&\{F,G\}^*_{\bar\varphi,\bar\pi,\delta\varphi,\delta\pi}.\nonumber
\eq
According to (\ref{DB_Def}), we have
\bq
\{F,G\}^*_{\bar\varphi,\bar\pi,\delta\varphi,\delta\pi}&=&\f{1}{V_0}\left(\f{\D
F}{\D \bar\varphi}\f{\D G}{\D \bar\pi}-\f{\D F}{\D \bar\pi}\f{\D G}{\D
\bar\varphi}\right)+\int{\md^3 x\left(\f{\delta F}{\delta(\delta
\varphi)}\f{\delta G}{\delta(\delta \pi)}\nonumber -\f{\delta
F}{\delta(\delta \pi)}\f{\delta G}{\delta(\delta
\varphi)}\right)}\nonumber\\
&-&\f{1}{V_0}\left(\int{\md^3 z\f{\delta F}{\delta (\delta
\varphi)}}\int{\md^3 z^\prime\f{\delta G}{\delta (\delta \pi)}}
-\int{\md^3 z\f{\delta F}{\delta (\delta \pi)}} \int{\md^3
z^\prime\f{\delta G}{\delta (\delta \varphi)}} \right)
 \nonumber\\
&=&\f{1}{V_0}\int{\md^3 z\f{\D f}{\D \varphi}}\int{\md^3
z^\prime\f{\D g}{\D \pi}} +\int{\md^3 x\f{\delta F}{\delta \varphi}}
\f{\delta G}{\delta \pi} - (\varphi \leftrightarrow \pi)
\nonumber \\
&-&\f{1}{V_0}\int{\md^3 z \left(\f{\D f}{\D \varphi}-\D_a\f{\D
f}{\D(\D_a \varphi)}\right)}\int{\md^3 z^\prime \left(\f{\D g}{\D
\pi}-\D_a\f{\D g}{\D(\D_a \pi)}\right)} - (\varphi \leftrightarrow
\pi)\nonumber\\
&=&\int{\md^3 x\f{\delta F}{\delta \varphi}} \f{\delta
G}{\delta \pi} - (\varphi \leftrightarrow \pi)
\equiv \{F,G\}_{\varphi,\pi}\,. \nonumber
\eq
\end{proof}

In the second equality, we have used
\[
\f{\D f}{\D \bar\varphi} = \f{\D f}{\D \varphi}, \quad \quad \f{\delta
F}{\delta(\delta \varphi)}=\f{\delta F}{\delta \varphi}=\f{\D f}{\D
\varphi}-\D_a\f{\D f}{\D (\D_a \varphi) }
\]
and dropped the surface integrals originating from integration of
the total divergence terms. It is now easy to see that if higher
derivative terms were present in the functionals, they would have
merely led to additional surface terms and would not have
affected the final conclusion.

Since linear functionals do not contribute to Dirac brackets, they can
be omitted, and one can restrict consideration to terms of the zeroth
and second order only. Moreover, for functionals of an even order, the
second term in the Dirac bracket (\ref{DB_Def}) vanishes, and one can
simply use the Poisson bracket (\ref{PB_full}).

 A somewhat similar consistency issue arises when it comes to
 equations of motion, generated e.g. by a (Hamiltonian) constraint
\be%
 H=\int{\!\md^3 x \,h(\varphi,\nabla\varphi, \pi, \nabla\pi)}.
\ee
There are again two approaches: i) either derive the equations of
motion for the original fields and then split into the background
and (linear) perturbation or ii) expand the constraint and obtain
separately equations of motion for the homogeneous and
inhomogeneous parts of the field. In other words, one needs to
compare
\be%
 \{\varphi,H\} \quad {\rm with} \quad \{\bar\varphi, H\}^* \,\,
{\rm and} \,\, \{\delta \varphi,
 H\}^*\,.
\ee
We start by noting that
\be\label{Full_EoM}%
 \dot \varphi = \{\varphi,H\}_{\varphi,\pi}=
\f{\delta H}{\delta \pi}= \f{\D h}{\D \pi} - \D_a
\f{\D h}{\D(\D_a \pi )},
\ee
whereas the equation of motion for the background field
\bq%
\dot{\bar\varphi} = \{\bar\varphi,H\}_{\bar\varphi,\bar\pi} &=&
\f{1}{V_0} \int{\md^3
x \f{\D h}{\D \bar\pi} = \f{1}{V_0} \int{\md^3 x \f{\D h}{\D
\pi}}}\nonumber
\\
&=&\f{1}{V_0} \int{ \md^3 x \left(\f{\D h}{\D \pi} - \D_a \f{\D
h}{\D(\D_a \pi)}\right)}
=\f{1}{V_0} \int{ \md^3 x \f{\delta h}{\delta \pi}}\nonumber
\eq
coincides with the background part of the equation of motion
(\ref{Full_EoM}) for the total field, $(\delta H/\delta
\pi)^{(0)}$. At the same time, the equation of motion for
the perturbation
\bq%
\delta\dot \varphi (x)=
\{\delta\varphi(x),H\}^*_{\delta\varphi,\delta\pi} &=& \f{\delta
H}{\delta (\delta \pi(x))}- \f{1}{V_0}\int{\md^3 y \f{\delta
H}{\delta (\delta\pi(y))}} \nonumber
\\
&=&\int{\md^3 y \left(\delta(x-y)-\f{1}{V_0}\right)\f{\delta
H}{\delta\pi(y)}}
=\left(\f{\delta H}{\delta \pi}\right)^{(1)}\nonumber
\eq
is nothing else but the perturbed part of the equation
(\ref{Full_EoM}). In fact, one can think of the kernel
$\delta(x-y)-1/V_0$ as cutting off the background
part of the function, with which it is integrated. It is again
pertinent to mention that the linear (as well as the background)
part of the functional does not contribute to the perturbed
equation of motion, that is
\be%
\delta\dot \varphi (x)=
\{\delta\varphi(x),H\}^*_{\delta\varphi,\delta\pi}=\{\delta\varphi(x),H^{(2)}\}_{\delta\varphi,\delta\pi}\,.
\ee
Note that in the second equality the Poisson bracket is used, not the
Dirac bracket.

To summarize, we have shown that in order to proceed to the
perturbation theory, the Dirac bracket (\ref{DB_Def}) in terms of
the background and perturbed variables should be used.
Nevertheless, when dealing with already expanded functionals,
containing only even order terms, the Dirac bracket reduces to the
Poisson bracket (\ref{PB_full}).

So far we have considered perturbations of a scalar field.
Generalization to tensorial fields is rather straightforward for any
rank.  In particular, we need the canonical pair of loop quantum
gravity, i.e. the extrinsic curvature and densitized triad whose
perturbations have
Dirac brackets
\be\label{DB_Def_KE}
\{F,G\}^*_{\delta K_a^i,\delta E_i^a}=
\{F,G\}_{\delta K_a^i,\delta E_i^a}-\f{1}{V_0}\left(\int{\md^3 z\md^3 z^\prime\f{\delta
F}{\delta (\delta K_a^i(z))}}{\f{\delta G}{\delta
(\delta E_i^a(z^\prime))}} - (F\leftrightarrow G)\right),
\ee
where $F$ and $G$ are arbitrary functionals of $K_a^i$ and $E_i^a$.

Of interest is also a generalization of the Dirac brackets to the case
of local second-class constraints. Let us split the triad and
extrinsic curvature into the diagonal and traceless parts
\begin{equation}
K_a^i=\kappa \delta_a^i + \kappa_a^i, \quad E^a_i=\varepsilon \delta^a_i + \varepsilon^a_i,
\end{equation}
such that
\be\label{2nd_class_constraints_local}
\chi_1:=\tr\varepsilon_i^a=0, \quad \chi_2:=\tr\kappa_a^i=0.
\ee It is easy to see that the pairs $(\kappa,\varepsilon)$ and
$(\kappa_a^i,\epsilon_i^a)$ are symplectically orthogonal. Indeed, the
symplectic structure takes the form
\[
\int{\md^3 x(\dot \varepsilon \delta_a^i + \dot\varepsilon_a^i)(\kappa\delta^i_a+\kappa^i_a)}
=\int{\md^3 x(3 \dot \varepsilon \kappa + \dot\varepsilon_a^i \kappa^i_a)}.
\]

However the constraints (\ref{2nd_class_constraints_local}) are second
class under the tentative Poisson bracket
\begin{equation}\label{PB_local}
\{F,G\}_{\kappa,\varepsilon,\kappa_a^i,\varepsilon_i^a}=\f{1}{3}\int{\md^3 x \left(\f{\delta F}{\delta\kappa}\f{\delta G}{\delta \varepsilon}-\f{\delta F}{\delta\varepsilon}\f{\delta G}{\delta \kappa} \right)}
+\int{\md^3 x\left(\f{\delta F}{\delta\kappa_a^i}\f{\delta G}{\delta \varepsilon_i^a}-\f{\delta F}{\delta\varepsilon_i^a}\f{\delta G}{\delta\kappa_a^i}\right)}.
\end{equation}
Specifically,
\begin{equation}
\{\chi_2(x),\chi_1(y)\}_{\kappa_a^i,\varepsilon_i^a}=3\delta(x-y).
\end{equation}
As before, we define the Dirac brackets
\begin{equation}
\{F,G\}^*=\{F,G\}-\int{\md^3 z \md^3 z^\prime \{F,\chi_a(z)\}C^{ab}(z,z^\prime)\{\chi_b(z),G\} },
\end{equation}
where the matrix $C^{ab}(x,y)$ is now space-dependent and satisfies
\begin{equation}
\int{\md^3 y C^{ab}(x,y)\{\chi_b(y),\chi_c(z)\} }=\delta_c^a\delta(x-z).
\end{equation}
Using the constraints (\ref{2nd_class_constraints_local}) in the equation above, we find that
\begin{eqnarray}
C^{11}(x,y)&=&C^{22}(x,y)=0 \\
C^{12}(x,y)&=&-C^{21}(x,y)=\f{1}{3}\delta(x-y).
\end{eqnarray}
Therefore the Dirac bracket reads
\be\label{DB_Def_local}
\{F,G\}^*_{\kappa_a^i,\varepsilon_i^a}=\int{\md^3 z \f{\delta F}{\delta \kappa_a^i(z)}\f{\delta G}{\delta\varepsilon_i^a(z)}}
-\f{1}{3}\int{\md^3 z \left(\f{\delta^j_b \delta F}{\delta \kappa_b^j(z)} \f{\delta^c_k\delta G}{\delta \varepsilon_k^c(z)} \right)}
- (F\leftrightarrow G)
\ee

It is easy to see by inspection that the constraints
(\ref{2nd_class_constraints_local}) indeed commute under this Dirac
bracket, $\{\chi_1,\chi_2\}^*=0$, so these constraints may be imposed
strongly. Also, the Dirac bracket between the original canonical
variables has the correct expression
\[
\{K^i_a(x),E^b_j(y)\}^*_{\kappa,\varepsilon,\kappa_c^k,\varepsilon_k^c}
=\delta_a^b \delta_j^i \delta(x-y)
=\{K^i_a(x),E^b_j(y)\}_{K_c^k,E^k_c}.
\]
Earlier on we have seen that one can still use the Poisson bracket
rather than the corresponding Dirac bracket if one removes from the
original constraints (that is before splitting the canonical
variables) all the terms proportional to the second class constraints
arising because of the splitting. This still holds for local second
class constraints. In the case at hand, as soon as all the terms
containing traces of the extrinsic curvature and the densitized triad
are omitted, the remaining constraints form the correct algebra under
the Poisson bracket (\ref{PB_local}).

\end{appendix}



\newcommand{\noopsort}[1]{}

\end{document}